\documentclass{article}

\usepackage{fullpage,bm,amsmath,amsfonts,amssymb,cite}
\usepackage{graphicx}
\usepackage{eepic, epsfig, latexsym, setspace, subfig}
\usepackage{color}
\usepackage{amsthm}

\newtheorem{theorem}{Theorem}[section]

\newcommand{\x}{\mathbf{x}}

\begin{document}

\title{Guarantees of Total Variation Minimization for Signal Recovery}
\author{Jian-Feng Cai and Weiyu Xu}

\maketitle

\begin{abstract}
In this paper, we consider using total variation minimization to recover signals whose gradients have a sparse support, from a small number of measurements. We establish the proof for the performance guarantee of total variation (TV) minimization in recovering \emph{one-dimensional} signal with sparse gradient support. This partially answers the open problem of proving the fidelity of total variation minimization in such a setting \cite{TVMulti}. In particular, we have shown that the recoverable gradient sparsity can grow linearly with the signal dimension when TV minimization is used. Recoverable sparsity thresholds of TV minimization are explicitly computed for $1$-dimensional signal by using the Grassmann angle framework. We also extend our results to TV minimization for multidimensional signals. Stability of recovering signal itself using $1$-D TV minimization has also been established through a property called ``almost Euclidean property for $1$-dimensional TV norm''. We further give a lower bound on the number of random Gaussian measurements for recovering $1$-dimensional signal vectors with $N$ elements and $K$-sparse gradients. Interestingly, the number of needed measurements is lower bounded by $\Omega((NK)^{\frac{1}{2}})$, rather than the $O(K\log(N/K))$ bound frequently appearing in recovering $K$-sparse signal vectors.

\end{abstract}

\section{Introduction}
Compressed sensing has recently gained a lot of attention in many applications including medical imaging, because it enables acquiring sparse signals from a much smaller number of samples than the ambient dimension of signal. Compressed sensing takes advantage of the fact that most signals of interest in practice are sparse: there are only a few nonzero or big elements when the signals are represented over a certain dictionary such as wavelet basis. For these types of sparse or compressible signals, compressed sensing theory \cite{CandesRombergTao,DonohoCS} has established that a small number of nonadaptive measurements are often sufficient to efficiently recover them under methods such as $\ell_1$ minimization \cite{CandesRombergTao,DonohoCS,CandesErrorCorrection, CT1}.

Without of loss of generality, let us assume that $\bm{x} \in \mathbb{R}^N $ is a one-dimensional (compared with $2$-dimensional images and $3$-dimensional videos) signal vector of $N$ elements, and has no more than $K$ ($K \ll N$) nonzero elements. In compressed sensing, we sample $\bm{x}$ using $M$ ($M<N$) linear projections
$$
  \bm{y}=A\bm{x},
$$
where $A$ is an $M \times N$ measurement matrix and $\bm{y}$ is an $M \times 1$ measurement result vector. Knowing $A$ and the measurement result $\bm{y}$, $\ell_1$ minimization is often used to recover the sparse $\bm{x}$:
\begin{equation}
       \min_{\bm{x}}{\|\bm{x}\|_1} ~~~~~\text{subject~to}~~\bm{y}=A\bm{x}.
\label{eq:L1program}
\end{equation}
It has been shown that under suitable conditions on the measurement matrix $A$, it is guaranteed that the original $\bm{x}$ is the unique solution to $\ell_1$ minimization \eqref{eq:L1program}. In fact, if $A$ satisfies the so-called restricted isometry property (RIP), then the solution of \eqref{eq:L1program} matches exactly with the original signal \cite{CS,CandesRombergTao, DonohoTanner}. Various results concerning the perfect reconstruction of the original signal by solving \eqref{eq:L1program} have been established in \cite{CandesRombergTao, DonohoCS, Neighborlypolytope,DonohoMalekiMontanari,DonohoTanner, Aspremont,  Kasin, GarnaevGluskin, XuHassibi,Yin}.

The results above hold true only for sparse signals, and they can be extended to signals that are synthesized by a linear combination of few atoms in a (redundant) dictionary with incoherent atoms \cite{CSRedundant}. However, there are numerous practical examples in which a signal of interest does not fall into the category where the aforementioned theory work. One such an example is signal that has a sparse gradient (i.e., the signal is piecewise constant), which arises frequently from imaging. Images with little detail are usually modelled as piecewise constant functions. For simplicity, we assume that $\bm{x}\in\mathbb{R}^N$ is a vector generated from $1$-dimensional piecewise constant signal. Let $D\bm{x}$ be its finite difference defined by $[D\bm{x}]_i=x_{i+1}-x_i$ for $i=1,2,\ldots,N-1$. Since $\bm{x}$ is piecewise constant, we must have that $D\bm{x}$ is sparse. Assume that $D\bm{x}$ has only $K$ ($K\ll N$) nonzero entries. Let $\bm{y}=A\bm{x}\in\mathbb{R}^M$ be $M$ linear samples of $\bm{x}$. Then, to recover $\bm{x}$, one usually solves
\begin{equation}\label{eq:TVprogram}
\min_{\bm{x}}{\|D\bm{x}\|_1} ~~~~~\text{subject~to}~~\bm{y}=A\bm{x}.
\end{equation}
The regularization term $\|D\bm{x}\|_1$ is called the \emph{total variation (TV)} of $\bm{x}$. When $\bm{x}\in\mathbb{R}^{N^d}$ is generated from $d$-dimensional signals, we only need to replace $D$ by the concatenation of directional finite differences, and $\|D\bm{x}\|_1$ is the anisotropic TV of $\bm{x}$.

TV regularization has been used extensively in the literature for decades in imaging sciences \cite{CDOS_JAMS,KeelingTV,OsherTV,CTTV} and other related fields \cite{TVSurface,TVProfile}. The minimization problem \eqref{eq:TVprogram} has the same form as the minimization in the analysis-based compressed sensing in \cite{CSAnalysis}. However, the perfect reconstruction result in \cite{CSAnalysis} can not be applied to \eqref{eq:TVprogram}, as the rows of $D$ do not form a frame ($D$ has a nontrivial null space). Despite the great importance of the TV minimization in applications, rigorous proofs of conditions of successfully recovering signal by using the TV minimization have only recently been established \cite{TVMulti,TV2d}. To establish such conditions, \cite{TVMulti,TV2d} first transformed $d$-dimensional ($d\geq2$) signals with sparse gradients into signals compressible over the Haar orthogonal wavelet basis. Then a modified restricted isometry condition, which takes into account the Haar orthogonal wavelet transformation, was established for the matrix $A$ such that \eqref{eq:TVprogram} offers a stable recovery of $\bm{x}$.  However, it is noted in \cite{TVMulti,TV2d} that establishing conditions for successfully recovering $1$-dimensional (namely $d$=1) signal vector remains an open problem.
%This is partially because $1$-dimensional signals with sparse gradients are not necessarily compressible over the Haar orthogonal basis.
This is partially due to the fact that small TV of a $1$-dimensional signal does not necessarily imply fast decay of its Haar wavelet coefficients.

In this paper, we establish the proof for performance guarantees of TV minimization in recovering \emph{$1$-dimensional} signal with sparse gradient support. This partially answers the open problem of proving the fidelity of total variation minimization in such a setting \cite{TVMulti}. Compared with \cite{TVMulti,TV2d}, our results do not use the restricted isometry condition, but instead directly work on the null space condition of the measurement matrix $A$.  To establish the null space condition of interest, we use ``Escape through the Mesh'' theorem \cite{Gordon,RV2008,StojnicThresholds} to estimate the Gaussian width \cite{Gordon, RV2008} of a cone specified by the null space condition. We then use the Grassmann angle framework to calculate the thresholds on gradient sparsity such that TV minimization \eqref{eq:TVprogram} can recover with high probability. We further extend our results to TV minimization for higher dimensional signals. For $d \geq 2$, we have obtained performance bounds for TV minimization comparable to results in \cite{TVMulti,TV2d}. We further give a lower bound on the number of random Gaussian measurements for recovering $1$-dimensional signal vectors with $N$ elements and $K$-sparse gradients. Interestingly, the number of needed measurements is lower bounded by $\Omega((NK)^{\frac{1}{2}})$, rather than the $O(K\log(N/K))$ bound frequently appearing in recovering $K$-sparse signal vectors. In \cite{MinimaxThresholds}, an average-case phase transition was calculated for TV minimization through evaluating the asymptotic minimax Mean Square Error (MSE) of TV minimization. Compared with \cite{MinimaxThresholds}, our results are more concerned with the worst-case performance guarantees which are uniformly true for all the possible supports for the signal gradient.

The rest of this paper is organized as follows. In Section \ref{sec:1d}, we establish the performance guarantee of TV minimization for $1$-dimensional signal vector. Our proof is based on a null space condition introduced in Section \ref{sec:Nullcondition}, and the derivations of performance guarantees of TV minimization are respectively presented in Section \ref{sec:escapefrommesh} and Section \ref{sec:lowerbound} using the escape through the mesh theorem and in Section \ref{sec:Grass} via Grassman angle framework.

In Section \ref{sec:multidimensional}, we extend our results to TV minimization for multidimensional signals. Section \ref{sec:conclusion} concludes our paper and discusses future directions.

\section{One-dimensional Signals}\label{sec:1d}
In this section, we establish the main result of this paper on the performance guarantee of TV minimization in recovering one-dimensional signal with sparse gradient support. Throughout this section, we will assume that $\bm{x}\in\mathbb{R}^N$ is generated from a one-dimensional signal and $D\bm{x}$ contains at most $K$ nonzeros. We assume that the entries of $A\in\mathbb{R}^{M\times N}$ are randomly drawn from i.i.d. Gaussian distribution. We give two different arguments on the proofs, namely, the one using ``Escape through the Mesh'' theorem \cite{Gordon,RV2008,StojnicThresholds} in Section \ref{sec:escapefrommesh} and the Grassmann angle framework in Section \ref{sec:Grass}. There two arguments will lead to two different recovery threshold bounds for minimal $M$. Both the arguments are based on a null space property of the matrix $A$, which are presented in Section \ref{sec:Nullcondition}.

\subsection{The Null Space Condition for Successful Recovery via the TV Minimization}
\label{sec:Nullcondition}

In this section, we give a condition on the null space of the linear projection matrix $A$, such that TV minimization successfully recovers one-dimensional signals with sparse gradients. We remark that this condition is not new, and  it has appeared in the proofs in \cite{TVMulti, TV2d}.

\begin{theorem} \label{thm:TVnullspacecondition1}
Assume $A\in\mathbb{R}^{M\times N}$ and $\bm{y}=A\bm{x}$.
%Suppose $\bm{x}$ is a signal vector whose gradients $D\bm{x}$ have no more than $K$ nonzero elements, $A$ is an $M\times N$ measurement matrix, and $\bm{y}=A\bm{x}$.
Then $\bm{x}$ is the unique solution to \eqref{eq:TVprogram} for all $\bm{x}$ whose gradients $D\bm{x}$ have no more than $K$ nonzero elements (no matter what the support $\mathcal{K}$ of $D\bm{x}$ is) if and only if the following condition holds: for every nonzero vector $\bm{z}$ in the null space of $A$ (namely $A\bm{z}=0$, $\bm{z} \neq \bm{0}$),
\begin{equation}\label{eq:TVnullspacecondition}
\|(D\bm{z})_{\mathcal{K}}\|_1< \|(D\bm{z})_{\mathcal{K}^c}\|_1 ~~\forall \mathcal{K}\subset\{1,2,\ldots,N-1\}\mbox{~s.t.~}|\mathcal{K}|\leq K.
\end{equation}
\end{theorem}

We omit the detailed proof of this theorem, since it is very similar to the proof of null space conditions for $\ell_1$ minimization; see, for example, \cite{StojnicThresholds, Yin}.

%This condition can also be easily extended to similar theorems for $2$-dimensional or higher dimensional signal vectors. We omit their corresponding statements and proofs here.

\subsection{Recovery Thresholds via Escape through the Mesh Theorem}
\label{sec:escapefrommesh}

In this subsection, we prove that a measurement matrix $A$ whose elements are i.i.d. Gaussian random variables satisfies the null space condition in Theorem \ref{thm:TVnullspacecondition1} with high probability, as long as
$$
M\geq C (NK)^{1/2}\ln N,
$$
%as $K \rightarrow \infty$ and $N \rightarrow \infty$,
where $C>0$ is a constant. Our proof builds on the following ``Escape through the Mesh'' theorem.

\begin{theorem} [Escape through the mesh \cite{Gordon}]
\label{thm:escapethroughmesh}
Let $\mathcal{S}$ be a subset of the unit Euclidean sphere $\mathbb{S}^{N-1}$ in $\mathbb{R}^N$. Let $Y$ be a random $(N-M)$-dimensional subspace of $\mathbb{R}^{N}$, distributed uniformly in the Grassmanian with respect to the Haar measure. Define the Gaussian width for the set $\mathcal{S}$ as $w(\mathcal{S})$=$E(\sup_{\bm{w}\in \mathcal{S}}$$(\bm{h}^T\bm{w}))$, where $\bm{h}$ is a random column vector in $\mathbb{R}^{N}$ with i.i.d. $\mathcal{N}(0,1)$ Gaussian elements. Assume that $w(\mathcal{S}) < (\sqrt{M}-\frac{1}{2\sqrt{M}})$. Then
\begin{equation*}
P(Y \bigcap \mathcal{S}=\emptyset)>1-3.5 e^{-\frac{(\sqrt{M}-\frac{1}{2\sqrt{M}})-w(\mathcal{S})}{18}}.
\end{equation*}
\end{theorem}

If the elements of the measurement matrix $A$  are i.i.d. Gaussian random variables, then the null space of $A$ is a random $(N-M)$-dimensional subspace distributed uniformly in the Grassmanian with respect to the Haar measure (see \cite{StojnicThresholds}). To prove the null space condition in Theorem \ref{thm:TVnullspacecondition1} holds with high probability, we show that the Gaussian width  $w(\mathcal{S})$ is in the order of $\sqrt{M}$ for the set
$$
\mathcal{S}=\{\bm{x}~:~\|\bm{x}\|_2= 1,\quad\mbox{and}\quad ~\exists \mathcal{K}\subset\{1,2,\ldots,N\}\mbox{~s.t.~}|\mathcal{K}|\leq K,~\|(D\bm{x})_{\mathcal{K}}\|_1\geq \|(D\bm{x})_{\mathcal{K}^c}\|_1\}.
$$

%Given any vector $\bm{x}$, we define $\|\bm{x}\|_{(K)}=\sqrt{\sum_{i=1}^{K}|x_i|^2}$ where $x_i$ is the $i$-th largest component of $\bm{x}$ in magnitude. It is obvious that $\|\bm{x}\|_{(K)}$ is indeed a norm, and its dual norm is defined as
%$$
%\|\bm{x}\|_{(K)}^*:=\max_{\|\bm{y}\|_{(K)}\leq 1}\langle \bm{x},\bm{y} \rangle.
%$$
%It is straightforward that
%$$
%\langle\bm{x},\bm{y}\rangle\leq \|\bm{x}\|_{(K)} \|\bm{y}\|_{(K)}^*.
%$$

For any $\bm{x}\in\mathcal{S}$ and a set $\mathcal{K}$ that satisfy $\|(D\bm{x})_{\mathcal{K}}\|_1\geq \|(D\bm{x})_{\mathcal{K}^c}\|_1$, we have that
$$
\|(D\bm{x})_{\mathcal{K}^c}\|_1\leq \|(D\bm{x})_{\mathcal{K}}\|_1\leq \sqrt{K}\|(D\bm{x})_{\mathcal{K}}\|_2
\leq\sqrt{K}\|D\bm{x}\|_2\leq 2\sqrt{K}\|\bm{x}\|_2=2\sqrt{K}.
$$
This implies $\|D\bm{x}\|_1\leq 4\sqrt{K}$ and therefore
%we partition $\bm{y}=D\bm{x}$ as
%$\bm{y}=[\bm{y}_0~\bm{y}_1~\ldots~\bm{y}_P]$, where $\bm{y}_0$ are the components of $\bm{y}$ on $\mathcal{K}$, $\bm{y}_1$ are the first $K$ largest components in magnitude on $\mathcal{K}^c$ and $\bm{y}_2$ are the next $K$ largest, and so on. Then, we have
%$\|\bm{y}_{i+1}\|_2\leq \sqrt{K(\|\bm{y}_i\|_1/K)^2} = \|\bm{y}_i\|_1/\sqrt{K}$. By summing it over all $i$, we get
%\begin{equation}\label{Eq:sumy2P}
%\sum_{i=2}^{P}\|\bm{y}_i\|_2\leq \frac{1}{\sqrt{K}}\sum_{i=1}^{P-1}\|\bm{y}_i\|_1\leq \frac{1}{\sqrt{K}}\|(D\bm{x})_{\mathcal{K}^c}\|_1
%\leq \frac{1}{\sqrt{K}}\|(D\bm{x})_{\mathcal{K}}\|_1
%\leq \|(D\bm{x})_{\mathcal{K}}\|_2\leq 2\|\bm{x}\|_2= 2.
%\end{equation}
%Also, we have $\|\bm{D}\bm{x}\|_2\leq 2\|\bm{x}\|_2=2$, which implies
%$\|\bm{y}_0\|_2^2+\|\bm{y}_1\|_2^2\leq 4$. Thus, for any vector $\bm{z}$ satisfying $\|\bm{z}\|_{(K)}\leq 1$, we partition it according to the $\bm{y}$ as $\bm{z}=[\bm{z}_0~\bm{z}_1~\ldots~\bm{z}_P]$ and it holds that $\|\bm{z}_i\|_2\leq 1$. Therefore,
%\begin{equation*}
%\begin{split}
%\langle D\bm{x},\bm{z} \rangle =\langle\bm{y},\bm{z}\rangle
%&=\langle\bm{y}_0,\bm{z}_0\rangle+\langle\bm{y}_1,\bm{z}_1\rangle+\sum_{i=2}^P\langle\bm{y}_i,\bm{z}_i\rangle\cr
%&\leq \sqrt{\|\bm{y}_0\|_2^2+\|\bm{y}_1\|_2^2}\cdot\sqrt{\|\bm{z}_0\|_2^2+\|\bm{z}_1\|_2^2}
%+ \sum_{i=2}^P\|\bm{y}_i\|_2\cdot\max_{2\leq i\leq P}\|\bm{z}_i\|_2\cr
%&\leq 2\sqrt{2} + 2.
%\end{split}
%\end{equation*}
%This implies that $\|D\bm{x}\|_{(K)}^* = \sup_{\|\bm{z}\|_{(K)}\leq 1}\langle D\bm{x},\bm{z}\rangle\leq 2\sqrt{2}+2$
and further
$$
\mathcal{S}\subset\widetilde{\mathcal{S}}:=\{\bm{x}~:~\|\bm{x}\|_2\leq 1,~~\|D\bm{x}\|_1\leq 4\sqrt{K}\}.
$$
In the following, we estimate the Gaussian width of $\widetilde{\mathcal{S}}$. We only consider the case that $N=2^L$, and the proof the other cases are essentially the same and does not change only the order the the Gaussian width.

For any $\bm{x}\in\widetilde{\mathcal{S}}$, we decompose $\bm{x}$ according to Haar wavelet transform as
\begin{equation}\label{Eq:Haar}
\bm{x} =\hat{\bm{z}}^{(1)}+\ldots+\hat{\bm{z}}^{(L)}+\hat{\bm{y}}^{(L)},
\end{equation}
where
$$
\hat{\bm{z}}^{(\ell)}=\bm{z}^{(\ell)}\otimes
[\underbrace{1~\ldots~1}_{2^{\ell-1}}~\underbrace{-1~\ldots~-1}_{2^{\ell-1}}],\quad
\bm{z}^{(\ell)}=[z_1^{(\ell)}~z_2^{(\ell)}~\ldots~z_{N/{2^\ell}}^{(\ell)}]
$$
and
$$
\hat{\bm{y}}^{(L)}=\bm{y}^{(L)}\otimes [1~1~\ldots~1],\quad \bm{y}^{(L)}=[y_1^{(L)}].
$$
Here $\otimes$ is the Kronecker product, i.e., $\bm{a}\otimes\bm{b}:=[a_1\bm{b}~a_2\bm{b}~\ldots~a_n\bm{b}]$. The decomposition \eqref{Eq:Haar} is done recursively as follows. We first decompose $\bm{x}=\hat{\bm{y}}^{(1)}+\hat{\bm{z}}^{(1)}$, where
$$
\hat{\bm{y}}^{(1)} = \bm{y}^{(1)}\otimes [1~1],
%[y_1^{(1)}~y_1^{(1)}~y_2^{(1)}~y_2^{(1)}~\ldots~\ldots~y_{N/2}^{(1)}~y_{N/2}^{(1)}],
\qquad \bm{y}^{(1)}=[y_1^{(1)}~y_2^{(1)}~\ldots~y_{N/2}^{(1)}],\quad y_i^{(1)} = \frac{x_{2i-1}+x_{2i}}{2},
$$
and
$$
\hat{\bm{z}}^{(1)} = \bm{z}^{(1)}\otimes [1~-1],\qquad \bm{z}^{(1)}=[z_1^{(1)}~z_2^{(1)}~\ldots~z_{N/2}^{(1)}],\quad
z_i^{(1)} = \frac{x_{2i-1}-x_{2i}}{2}.
$$
Then, we further decompose
$$
\hat{\bm{y}}^{(1)}=\hat{\bm{y}}^{(2)}+\hat{\bm{z}}^{(2)},
$$
where
$$
\hat{\bm{y}}^{(2)} = \bm{y}^{(2)}\otimes [1~1~1~1],\qquad \bm{y}^{(2)}=[y_1^{(2)}~y_2^{(2)}~\ldots~y_{N/4}^{(2)}],\quad
y_i^{(2)} = \frac{y_{2i-1}^{(1)}+y_{2i}^{(1)}}{2},
$$
and
$$
\hat{\bm{z}}^{(2)} = \bm{z}^{(2)}\otimes [1~1~-1~-1],\qquad \bm{z}^{(2)}=[z_1^{(2)}~z_2^{(2)}~\ldots~z_{N/4}^{(2)}],\quad
z_i^{(2)} = \frac{y_{2i-1}^{(1)}-y_{2i}^{(1)}}{2}.
$$
Generally, at level $\ell$, we have that
$$
\hat{\bm{y}}^{(\ell)}=\bm{y}^{(\ell)}\otimes
[\underbrace{1~\ldots~1}_{2^{\ell}}], \qquad \bm{y}^{(\ell)}=[y_1^{(\ell)}~y_2^{(\ell)}~\ldots~y_{N/{2^\ell}}^{(\ell)}]
$$
we decompose it as
$$
\hat{\bm{y}}^{(\ell)}=\hat{\bm{y}}^{(\ell+1)}+\hat{\bm{z}}^{(\ell+1)},
$$
where
$$
\hat{\bm{y}}^{(\ell+1)}=\bm{y}^{(\ell+1)}\otimes
[\underbrace{1~\ldots~1}_{2^{\ell+1}}],\qquad \bm{y}^{(\ell+1)}=[y_1^{(\ell+1)}~z_2^{(\ell+1)}~\ldots~y_{N/{2^{\ell+1}}}^{(\ell+1)}],\quad
y_i^{(\ell+1)} = \frac{y_{2i-1}^{(\ell)}+y_{2i}^{(\ell)}}{2},
$$
and
$$
\hat{\bm{z}}^{(\ell+1)}=\bm{z}^{(\ell+1)}\otimes
[\underbrace{1~\ldots~1}_{2^{\ell}}~\underbrace{-1~\ldots~-1}_{2^{\ell}}],\qquad \bm{z}^{(\ell+1)}=[z_1^{(\ell+1)}~z_2^{(\ell+1)}~\ldots~z_{N/{2^{\ell+1}}}^{(\ell+1)}],\quad
z_i^{(\ell+1)} = \frac{y_{2i-1}^{(\ell)}-y_{2i}^{(\ell)}}{2}.
$$

The decomposition \eqref{Eq:Haar} has the following properties.
\begin{itemize}
\item Obviously, components in decomposition \eqref{Eq:Haar} are orthogonal to each others. Consequently,
$$
\|\bm{x}\|_2^2 = \|\hat{\bm{z}}^{(1)}\|_2^2 + \|\hat{\bm{z}}^{(2)}\|_2^2 +\ldots + \|\hat{\bm{z}}^{(L)}\|_2^2 + \|\hat{\bm{y}}^{(L)}\|_2^2 =
\sum_{\ell=1}^{L}\left(2^{\ell}\|\bm{z}^{(\ell)}\|_2^2\right) + 2^L\|\bm{y}\|_2^2.
$$
Since $\bm{x}\in\tilde{\mathcal{S}}$ implies $\|\bm{x}\|_2^2\leq 1$, we have
\begin{equation}\label{Eq:Prop1}
\sum_{\ell=1}^{L}\left(2^{\ell}\|\bm{z}^{(\ell)}\|_2^2\right) + 2^L\|\bm{y}^{(L)}\|_2^2\leq 1.
\end{equation}
\item It can be shown that
$$
\|D\hat{\bm{y}}^{(\ell)}\|_1\leq \|D\hat{\bm{y}}^{(\ell-1)}\|_1
$$
and, therefore,
\begin{equation}\label{Eq:Prop2}
\|\bm{z}^{(\ell)}\|_1\leq \|D\hat{\bm{y}}^{(\ell-1)}\|_1/2\leq 2\sqrt{K}.
\end{equation}
Indeed, let $\bm{u}$ be satisfying $\|D\hat{\bm{y}}^{(\ell)}\|_1=\langle \bm{u},D\hat{\bm{y}}^{(\ell)}\rangle$ and $\|\bm{u}\|_\infty\leq 1$, and we then have
\begin{equation*}
\begin{split}
\langle\bm{u},D\hat{\bm{y}}^{(\ell)}\rangle
&= \sum_{i=1}^{N/2^{\ell}-1} \left(u_{i2^\ell}\cdot(y_{i+1}^{(\ell)}-y_{i}^{(\ell)})\right)
=\sum_{i=1}^{N/2^{\ell}-1}\left((u_{i2^{\ell}})\cdot\left(\frac{y_{2i+2}^{(\ell-1)}+y_{2i+1}^{(\ell-1)}}{2}
-\frac{y_{2i}^{(\ell-1)}+y_{2i-1}^{(\ell-1)}}{2}\right)\right)\cr
&=\sum_{i=1}^{N/2^{\ell}-1}\left(u_{i2^{\ell}}\cdot\left(\frac{y_{2i+2}^{(\ell-1)}-y_{2i+1}^{(\ell-1)}}{2}
+(y_{2i+1}^{(\ell-1)}-y_{2i}^{(\ell-1)})
+\frac{y_{2i}^{(\ell-1)}-y_{2i-1}^{(\ell-1)}}{2}\right)\right)\cr
&=\sum_{i=1}^{N/2^{\ell}-1}\left(u_{i2^{\ell}}\cdot\left(\frac{y_{2i+2}^{(\ell-1)}-y_{2i+1}^{(\ell-1)}}{2}\right)\right)
+\sum_{i=1}^{N/2^{\ell}-1}\left(u_{i2^{\ell}}\cdot(y_{2i+1}^{(\ell-1)}-y_{2i}^{(\ell-1)})\right)\cr
&\qquad +\sum_{i=1}^{N/2^{\ell}-1}\left(u_{i2^{\ell}}\cdot\left(\frac{y_{2i}^{(\ell-1)}-y_{2i-1}^{(\ell-1)}}{2}\right)\right)\cr
&=\frac{u_{2^{\ell}}}{2}\cdot\left(y_{2}^{(\ell-1)}-y_{1}^{(\ell-1)}\right)+
\sum_{i=2}^{N/2^{\ell}-1}\left(\frac{u_{i2^{\ell}}+u_{(i-1)2^{\ell}}}{2}\cdot\left(y_{2i}^{(\ell-1)}-y_{2i-1}^{(\ell-1)}\right)\right)\cr
&\qquad+\frac{u_{N-2^{\ell}}}{2}\cdot\left(y_{N/2^{\ell-1}}^{(\ell-1)}-y_{N/2^{\ell-1}-1}^{(\ell-1)}\right)
+\sum_{i=1}^{N/2^{\ell}-1}\left(u_{i2^{\ell}}\cdot(y_{2i+2}^{(\ell-1)}-y_{2i+1}^{(\ell-1)})\right)\cr
& = \langle \tilde{\bm{u}}, D\hat{\bm{y}}^{(\ell-1)} \rangle,
\end{split}
\end{equation*}
where
$$
\tilde{\bm{u}}=\left[\bm{0}~\frac{u_{2^\ell}}{2}~\bm{0}~u_{2^\ell}~\bm{0}
~\frac{u_{22^{\ell}}+u_{2^{\ell}}}{2}~\bm{0}~u_{22^\ell}~\bm{0}
~\frac{u_{32^{\ell}}+u_{22^{\ell}}}{2}~\bm{0}~u_{32^\ell}~\bm{0}~
\ldots\ldots\ldots~\bm{0}
~\frac{u_{N-2^{\ell}}}{2}~\bm{0}\right].
$$
Since
$$
\left|\frac{u_{i2^{\ell}}+u_{(i-1)2^{\ell}}}{2}\right|
\leq \frac{|u_{i2^{\ell}}|+|u_{(i-1)2^{\ell}}|}{2},
$$
we have $\|\tilde{\bm{u}}\|_\infty\leq\|\bm{u}\|_\infty\leq 1$. This leads to
$$
\|D\hat{\bm{y}}^{(\ell)}\|_1=\langle \bm{u},D\hat{\bm{y}}^{(\ell)}\rangle
=\langle \tilde{\bm{u}},D\hat{\bm{y}}^{(\ell-1)}\rangle
\leq \|D\hat{\bm{y}}^{(\ell-1)}\|_1.
$$
\end{itemize}

Now we are ready to estimate the Gaussian width of $\tilde{\mathcal{S}}$. Let $\bm{g}$ be a vector whose entries are i.i.d. Gaussian random variables with mean $0$ and variance $1$. Since \eqref{Eq:Prop1} implies $\|\bm{z}^{(\ell)}\|_2\leq \frac{1}{\sqrt{2^{\ell}}}$, we have, by Cauchy-Schwartz inequality,
$
\|\bm{z}^{(\ell)}\|_1
%=\langle \bm{u},\bm{z}^{\ell}\rangle
%\leq \sum_{i=1}^{\lceil(N/2^{\ell})/K\rceil}\|\bm{u}_i\|_2\|\bm{z}^{\ell}_i\|_2
%\leq \sum_{i=1}^{\lceil(N/2^{\ell})/K\rceil}\|\bm{z}^{\ell}_i\|_2
%\leq \sqrt{\frac{N/K}{2^{\ell}}+1}\|\bm{z}^{(\ell)}\|_2\leq\frac{\sqrt{N/K}}{2^{\ell}}+\frac{1}{\sqrt{2^{\ell}}}.
\leq \sqrt{\frac{N}{2^\ell}}\|\bm{z}^{(\ell)}\|_2\leq \frac{\sqrt N}{2^\ell}$.
This together with \eqref{Eq:Prop2} implies that
$$
\|\bm{z}^{(\ell)}\|_1\leq \min\left\{\frac{\sqrt{N}}{2^{\ell}}, 2\sqrt{K} \right\}.
$$
Then,
$$
\langle \hat{\bm{z}}^{(\ell)},\bm{g} \rangle
=\langle \bm{z}^{(\ell)},\bm{g}^{(\ell)} \rangle
\leq \|\bm{z}^{(\ell)}\|_1\|\bm{g}^{(\ell)}\|_\infty.
$$
Here
$$
\bm{g}^{(\ell)}=\left[\sum_{i=1}^{2^{\ell-1}}(g_i-g_{i+2^{\ell-1}})~\sum_{i=1}^{2^{\ell-1}}(g_{i+2^\ell}-g_{i+2^\ell+2^{\ell-1}})
~\ldots~\sum_{i=1}^{2^{\ell-1}}(g_{i+N-2^\ell}-g_{i+N-2^\ell+2^{\ell-1}})\right]
:=[g_1^{(\ell)}~g_2^{(\ell)}~\ldots~g_{N/2^{\ell}}^{(\ell)}]
$$

In the following, we estimate $E(\|\bm{g}^{(\ell)}\|_\infty)$. Notice that the components in $\bm{g}^{(\ell)}$ are i.i.d. random variables that follow $\mathcal{N}(0,2^\ell)$.
%$$
%E(\max_{0\leq i \leq N/2^{\ell}-1}|g_{2^{\ell}i+1}+\ldots+g_{2^{\ell}i+2^{\ell-1}}-g_{2^{\ell}i+2^{\ell-1}+1}-\ldots-g_{2^{\ell}(i+1)}|) = E(\max_{1\leq i\leq N/2^{\ell}}|\tilde{g}_i|),\quad \tilde{g}_i\sim\mathcal{N}(0,2^\ell).
%$$
The following argument follows from Lemma 4.4 of Rudelson and Vershynin's paper \cite{RV2008}. Let $p$ be a large enough number that is determined later. Then
\begin{equation*}
\begin{split}
E\left(\|\bm{g}^{(\ell)}\|_\infty\right)&\leq E\left(\left(\sum_{i}|g_i^{(\ell)}|^p\right)^{1/p}\right)
\leq \left(N/2^{\ell}\right)^{1/p} \left(E\left(|{g}_i^{(\ell)}|^p\right)\right)^{1/p}\cr
&\leq \sqrt{2^{\ell}} \left(N/2^{\ell}\right)^{1/p}\left(2^{p/2}\frac{\Gamma(p/2+1/2)}{\Gamma(1/2)}\right)^{1/p}
\cr
&\leq \sqrt{2^{\ell}}\left(N/2^{\ell}\right)^{1/p}\left(\frac{p+1}{e}\right)^{1/2}(1+o(1)).%\cr
%&= (N/2^{\ell})^{1/p} \left( 2^{p\ell/2}\frac{2^{p/2}\Gamma\left(\frac{p+1}{2}\right)}{\sqrt{\pi}}\right)^{1/p}
%= \sqrt{2^\ell}(N/2^{\ell})^{1/p} \left(\frac{2^{p/2}\Gamma\left(\frac{p+1}{2}\right)}{\sqrt{\pi}}\right)^{1/p}\cr
%%&\leq \sqrt{2^\ell}(N/2^{\ell})^{1/p} \left(\frac{2^{p/2}\sqrt{\frac{2\pi}{\frac{1+p}{2}}}\left(\frac{p+1}{2e}\right)^{(p+1)/2}(1+o(2/(p+1)))}{\sqrt{\pi}}\right)^{1/p}
%&\leq \sqrt{2^\ell}(N/2^{\ell})^{1/p} \left(\frac{2^{p/2}e\left(\frac{1+p}{2e}\right)^{p/2}\sqrt{e}}{\sqrt{\pi}}\right)^{1/p}\cr
%&= \sqrt{2^\ell}\left(\frac{N}{2^{\ell}}\frac{e^{3/2}}{\sqrt{\pi}}\right)^{1/p} \left(\frac{1+p}{e}\right)^{1/2}\cr
\end{split}
\end{equation*}
Choose $p=2\ln\left(\frac{N/2^{\ell}}{K}\right)$, and we obtain
$$
E\left(\|\bm{g}^{(\ell)}\|_\infty\right)\leq
\sqrt{2^\ell} \left(p+1\right)^{1/2} =\sqrt{2^{\ell}}\sqrt{1+2\ln\left(N/2^{\ell}\right)}
=\sqrt{2^{\ell}}\sqrt{2\ln\left(e^{1/2}N/2^{\ell}\right)}
$$

Therefore,
\begin{equation}\label{Eq:Supz}
\begin{split}
E\left(\sup_{\bm{x}\in\tilde{\mathcal{S}}}\langle \hat{\bm{z}}^{\ell},\bm{g} \rangle\right)
&\leq
E\left(\sup_{\bm{x}\in\tilde{\mathcal{S}}}\|\hat{\bm{z}}^{\ell}\|_1\|\bm{g}^{(\ell)}\|_\infty\right)
\leq
\min\left\{\frac{\sqrt{N}}{2^{\ell}}, 2\sqrt{K}\right\}\cdot E\left(\|\bm{g}^{(\ell)}\|_{(K)}\right)\cr
&\leq \min\left\{\sqrt{\frac{N}{2^{\ell}}}, 2\sqrt{2^{\ell} K}\right\}\cdot\sqrt{2\ln\left(e^{1/2}N/2^{\ell}\right)}.
\end{split}
\end{equation}
%Here we have used the fact that
%$$
%E\left(\left|\sum_{i=1}^Mg_i\right|\right) = \sqrt{\frac{2M}{\pi}}
%$$
%and
%$$
%E(\max_{1\leq i\leq M}|G_i|)\leq \ln M\cdot E(|G_i|)
%$$
%for i.i.d. random variables $G_i$.
%
%Furthermore, since $2^{L}\|\bm{y}^{(L)}\|_2^2\leq 1$, we have $|y_1^{(L)}|\leq 1/\sqrt{2^L}=1/\sqrt{N}$ and then
%$$
%E\left(\sup_{\bm{x}\in\tilde{\mathcal{S}}}\langle \hat{\bm{y}}^{(L)},\bm{g} \rangle\right)
%\leq  E\left(\sup_{\bm{x}\in\tilde{\mathcal{S}}} |y_1^{(L)}|\cdot\left|\sum_{i=1}^N g_i\right| \rangle\right)\leq 1/\sqrt{N} E\left(\left|\sum_{i=1}^N g_i\right| \rangle\right) = 1/\sqrt{N}\sqrt{\frac{2N}{\pi}} = \sqrt{\frac{2}{\pi}}.
%$$
This together with \eqref{Eq:Haar} implies that
$$
E_{\bm{g}}\left(\sup_{\bm{x}\in\tilde{\mathcal{S}}}\langle\bm{x},\bm{g}\rangle\right)
\leq \sqrt{\frac{2}{\pi}} + \sum_{\ell=1}^{L}\min\left\{\sqrt{\frac{N}{2^{\ell}}}, 2\sqrt{2^{\ell} K}\right\}\cdot\sqrt{2\ln\left(e^{1/2}N/2^{\ell}\right)}
$$
Now we estimate the constant $\sum_{\ell=1}^{L}\min\left\{\sqrt{\frac{N}{2^{\ell}}}, 2\sqrt{2^{\ell} K}\right\}\cdot\sqrt{2\ln\left(e^{1/2}N/2^{\ell}\right)}$. Let $L_0$ be the maximum integer that satisfies
$\sqrt{\frac{N}{2^{L_0}}}\geq2\sqrt{2^{L_0}K}$, which leads to $2^{L_0} \leq \frac{1}{2}\sqrt{\frac{N}{K}}$. Since $L_0$ is the maximum integer, we have
$\frac{1}{2}\sqrt{\frac{N}{K}}\leq 2^{L_0+1}$.
It is obviously that $\min\left\{\sqrt{\frac{N}{2^{\ell}}}, 2\sqrt{2^{\ell}K}\right\}=2\sqrt{2^{\ell}K}$ if $\ell\leq L_0$ and $\min\left\{\sqrt{\frac{N}{2^{\ell}}}, 2\sqrt{2^{\ell}K}\right\}=\sqrt{\frac{N}{2^{\ell}}}$ otherwise. Therefore, if $N>1$ and $K>1$, then
\begin{equation}
\begin{split}
&\sum_{\ell=1}^{L}\left(\min\left\{\sqrt{\frac{N}{2^{\ell}}}, 2\sqrt{2^{\ell}K}\right\}\cdot \sqrt{2\ln\left(e^{1/2}N/2^{\ell}\right)}\right)\cr
\leq&\sqrt{2\ln\left(e^{1/2}N\right)} \left(2\sqrt{K}\sum_{\ell=1}^{L_0}\sqrt2^{\ell}+\sqrt{N}\sum_{\ell=L_0+1}^{L}\left(\frac{1}{\sqrt2}\right)^{\ell}\right)\cr
=& \sqrt{2\ln\left(e^{1/2}N\right)}  \left(2\sqrt2\sqrt{K}\frac{\sqrt{2}^{L_0}-1}{\sqrt2-1} +\sqrt{N}\left(\frac{1}{\sqrt2}\right)^{L_0+1}\frac{1-\left(\frac{1}{\sqrt2}\right)^{L-L_0}}{1-\frac{1}{\sqrt2}}\right)\cr
\leq & \sqrt{2\ln\left(e^{1/2}N\right)} \left(\left(\frac{2\sqrt2}{\sqrt2-1}\right)\sqrt{K}\sqrt{2}^{L_0}+\sqrt{\frac{N}{2^{L_0+1}}}\frac{1}{1-\frac{1}{\sqrt2}}\right)-\sqrt{\frac{2}{\pi}}\cr
\leq &\sqrt{2\ln\left(e^{1/2}N\right)}\left((4+2\sqrt2)\sqrt{K}\sqrt{\frac{1}{2}\sqrt{\frac{N}{K}}}+(2+\sqrt2)\sqrt{\frac{N}{\frac{1}{2}\sqrt{\frac{N}{K}}}}\right)-\sqrt{\frac{2}{\pi}}\cr
\leq &\left(4\sqrt2+4\right)(NK)^{1/4}\sqrt{2\ln\left(e^{1/2}N\right)} - \sqrt{\frac{2}{\pi}}
\end{split}
\end{equation}

Finally, we get
\begin{equation*}
\begin{split}
E\left(\sup_{\bm{x}\in\tilde{\mathcal{S}}}\langle \bm{x},\bm{g} \rangle\right)
=(4\sqrt2+4)(NK)^{1/4}\sqrt{2\ln\left(e^{1/2}N\right)}.
\end{split}
\end{equation*}
Since the Gaussian width in Theorem \ref{thm:escapethroughmesh} is of the order $\sqrt{M}$, we have
$$
M\sim (NK)^{1/2}\ln N.
$$

\section{Lower bound on the number of measurements for $1$-dimensional signal vector}
\label{sec:lowerbound}

\begin{theorem}
The number of random Gaussian measurements needed to guarantee, with high probability, the null space condition \ref{thm:TVnullspacecondition1} for signals with $K$-sparse gradient  is at least $\frac{\pi}{16} (NK)^{\frac{1}{2}}-O(\sqrt{N})$.
\label{thm:lowerbound_numbermeasurement}
\end{theorem}

We will prove, when $K$ and $N$ are large enough,  the Gaussian width is lower bounded as
$$
E\left(\sup_{\bm{x}\in\mathcal{S}}<\bm{x},\bm{g}>\right)\geq \frac{\sqrt{\pi}}{4}\left(NK\right)^{1/4},
$$
where $\bm{g}\sim \mathcal{N}(0,I)$ and
$$
\mathcal{S}=\{\bm{x}~|~\|\bm{x}\|_2\leq 1~\mbox{and}~\exists \mathcal{K}\subset\{1,\ldots,N-1\}\mbox{~s.t.~}|
\mathcal{K}|\leq K,~\|(D\bm{x})_{\mathcal{K}}\|_1\geq\|(D\bm{x})_{\mathcal{K}^c}\|_1\}.
$$

We assume $K\leq N/3-1$.
Let $\bm{g}\in\mathbb{R}^{N}$ be fixed. Let $L$ be positive number that is to be determined later. We partition $\bm{g}$ as
$$
\bm{g}=[\bm{g}_1,\ldots,\bm{g}_{H},\bm{g}_{H+1}],
$$
where $H=\lfloor(N-\max\{K+1,L\})/L\rfloor$ and $\bm{g}_1,\ldots,\bm{g}_H\in\mathbb{R}^{L}$ and $\bm{g}_{H+1}$ is the remaining entries of $\bm{g}$.

Let $\bm{a}$ be a vector that has the same size as $\bm{g}_{H+1}$. Since the length of $\bm{g}_{H+1}$ is larger than $K+1$, we can define
$$
\bm{a}=[0,\ldots,0,\underbrace{1,-1,1,-1,\ldots,1,-1}_{K~\mathrm{terms}}]
$$
Therefore, the support $\mathcal{K}_0$ of $D\bm{a}$ has a cardinality $K$.

We define
$$
\tilde{g}=[\langle\bm{g}_1,\bm{1}\rangle~\ldots~\langle\bm{g}_{H},\bm{1}\rangle,\langle\bm{g}_{H+1},\bm{a}\rangle]^T
$$
and
$$
\bm{s}=\mathrm{sgn}(\tilde{\bm{g}})
$$
Let
$$
\bm{x}=[\nu s_1\bm{1}~\ldots~\nu s_{H}\bm{1}~\mu s_{H+1}\bm{a}],
$$
where $\mu$ and $\nu$ are positive numbers to be determined later.

Therefore, we have $\|\bm{x}\|_2\leq (\nu^2 N+\mu^2 K)^{1/2}$. In order $\|\bm{x}\|_2\leq 1$, we need
\begin{equation}\label{eq:l2cons}
\nu^2 N + \mu^2 K\leq 1.
\end{equation}
Moreover, from the construction, $\|(D\bm{x})_{\mathcal{K}_0}\|_1=(2K-1)\mu$, and $\|(D\bm{x})_{\mathcal{K}_0^c}\|_1\leq(2H-1)\nu\leq (2N/L)\nu$. In order that $\bm{x}\in\mathcal{S}$, we should have
\begin{equation}\label{eq:l1cons}
(2K-1)\mu\geq 2\nu N/L.
\end{equation}

We pick
$$
\mu=\frac{1}{\sqrt{2K}},\quad
\nu=\frac{1}{\sqrt{2N}},\quad
L=\sqrt{\frac{4NK}{(2K-1)^2}}.
$$
%$$
%\mu=\frac{1}{2\sqrt{K}},\quad
%\nu=\frac{1}{2\sqrt{N}},\quad
%L=\frac12\sqrt{\frac{N}{K}}.
%$$
Then, both \eqref{eq:l2cons} and \eqref{eq:l1cons} are satisfied, and hence $\bm{x}\in\mathcal{S}$. %Since we assume $K\leq N/3-1$, we have $H-1\geq \frac{4}{3}\sqrt{NK}-2\geq\sqrt{NK}$ if $NK\geq 81$.
Since we assume $K\leq N/3-1$, for any constant $\epsilon>0$, we have $H-1\geq \frac{2}{3} (1-\epsilon) \sqrt{NK}$ when $K$ and $N$ are large enough. For large enough $K$ and $N$, we finally have
$$
\langle\bm{x},\bm{g}\rangle = \mu|\tilde{g}_1|+\nu\sum_{i=2}^{N/L}|\tilde{g}_i|
$$
and thus
\begin{equation*}
\begin{split}
E\left(\sup_{\bm{x}\in\mathcal{S}}<\bm{x},\bm{g}>\right)&\geq
E\left( \mu|\tilde{g}_1|+\nu\sum_{i=2}^{N/L}|\tilde{g}_i|\right)
=\mu E(|\tilde{g}_1|)+\nu(H-1)E(|\tilde{g}_2|)\cr
%&\geq \nu(H-1)E(|\tilde{g}_2|)
%\geq \frac{1}{2\sqrt{N}}\cdot\sqrt{NK}\cdot\sqrt{\frac12\sqrt{\frac{N}{K}}}\cdot\sqrt{\pi/2}
%=\frac{\sqrt{\pi}}{4}\cdot(NK)^{1/4}.
&\geq \nu(H-1)E(|\tilde{g}_2|)
\geq \frac{1}{\sqrt{2N}}\cdot\frac{2}{3}(1-\epsilon)\sqrt{NK}\cdot\sqrt{\sqrt{\frac{4NK}{(2K-1)^2}}}\cdot\sqrt{\pi/2}\cr
&=\frac{(1-\epsilon)\sqrt{\pi}}{3} \left( \frac{4NK^3}{(2K-1)^2} \right)^{\frac{1}{4}}.
\end{split}
\end{equation*}
This is bigger than $\frac{\pi}{4} (NK)^{\frac{1}{4}}$ when $K$ and $N$ are large enough.

We now show that the number of random Gaussian measurements needed to guarantee, with high probability, the null space condition \ref{thm:TVnullspacecondition1} for signals with $K$-sparse gradient  is at least $\frac{\pi}{16} (NK)^{\frac{1}{2}}-O(\sqrt{N})$.

Let us consider one signal vector $\bm(z)$ whose last $K$ elements are $(-1,+1,-1,+1,\ldots,-1,+1)$, and other $(N-K)$ elements are zero. Then the descent cone $DC$ for $\bm{z}$ under TV minimization is a subset of the non-convex cone
$$
\mathcal{S}=\{\bm{x}~|~\exists \mathcal{K}\subset\{1,\ldots,N-1\}\mbox{~s.t.~}|
\mathcal{K}|\leq K,~\|(D\bm{x})_{\mathcal{K}}\|_1\geq\|(D\bm{x})_{\mathcal{K}^c}\|_1\}.
$$

Using the same derivation technique as above in this section, we have that the Gaussian width $w(DC)$ for the descent cone $DC$ is at least $\frac{\sqrt{\pi}}{4}\left(NK\right)^{1/4}$. Notice that the descent cone is a convex one, and a subset of the nonconvex cone
$$
\mathcal{S}=\{\bm{z}~|~\exists \mathcal{K}\subset\{1,\ldots,N-1\}\mbox{~s.t.~}|
\mathcal{K}|\leq K,~\|(D\bm{x})_{\mathcal{K}}\|_1\geq\|(D\bm{x})_{\mathcal{K}^c}\|_1\}.
$$

From Proposition 10.1 in \cite{TroppEdge}, the statistical dimension $\delta(DC)$ of the descent cone is at least $w^2(DC)$, which is $\frac{\pi}{16} (NK)^{\frac{1}{2}}$. From Theorem II in \cite{TroppEdge}, if the number of measurements $M \leq \delta(DC)-4\sqrt{\log(4/\eta)} \sqrt{N}$, then TV minimization succeeds with probability no bigger than $\eta$.  If we take $\eta$ small, with high probability, TV minimization fails to recover a signal vector with $K$-sparse gradient with high probability, if the number of measurements $M \leq \frac{\pi}{16} (NK)^{\frac{1}{2}}-4\sqrt{\log(4/\eta)} \sqrt{N}$.

In summary, we have the theorem concerning the lower bound on the number of measurements for TV minimization using random Gaussian matrices.

\subsection{Recovery Thresholds via the Grassmann Angle Framework}
\label{sec:Grass}

In previous subsections, we have used the ``Escape through the Mesh'' theorem to establish performance guarantees of TV minimization for signal recovery. In this subsection, we explore the Grassmann angle framework \cite{XuHassibi} to characterize performance guarantees of TV minimization for $1$-dimensional signal vectors. The upshot here is that the Grassmann angle framework gives explicitly computable thresholds on recoverable sparsity level $K$, when the number of measurements is proportionally growing with the signal dimension $N$.

Let us use $\mathcal{K}$ to denote the set of indices $i$'s such that $|x_{i+1}-x_i|$ is one of the $K$ terms on the left side of the inequality \ref{thm:TVnullspacecondition1}. Let us denote the set of indices ($i$'s and $(i+1)$'s ) involved in these $K$ terms as $\mathcal{DK}$. We note that the cardinality $|\mathcal{DK}|$ of $\mathcal{DK}$ is at most $2K$.

Then there exist at least $(N-1-3K)$ terms in the form of $|x_{i+1}-x_i|$ that do not involve any index in $\mathcal{DK}$. Among these $(N-1-3K)$ terms, we can at least choose $\frac{N-1-3K}{2}$ terms such that each of them involves different indices from $\mathcal{DK}$, {\it and} from each other. Let us use $\mathcal{KB}$ to denote the set of indices $i$'s such that $|x_{i+1}-x_i|$ is one of these $\frac{N-1-3K}{2}$ terms. By the triangle inequality,
$$
 \sum_{i \in \mathcal{K}} |x_{i+1}-x_i| \leq  2 \sum_{i \in \mathcal{DK}}|x_i|
$$
Then one sufficient condition for $TV$ minimization to work is
\begin{equation}\label{eq:relaxedcond}
 2 \sum_{i \in \mathcal{DK}}|x_i| \leq \sum_{i \in \mathcal{KB}} |x_{i+1}-x_i|
\end{equation}
holds for every vector $\bm{x}$ in the null space of the projection $A$. We call this condition $RelaxedNULL$ condition.

Since we are taking the projection $A$ uniformly over all the $m$-dimensional subspaces in $R^{n}$, the probability that $RelaxedNULL$ condition holds, is equivalent to the probability that
\begin{equation}\label{eq:chdimcond}
 2 \sum_{i \in \{1,2,..., |\mathcal{DK}|\}}|x_i| \leq \sqrt{2}\sum_{i \in \{|\mathcal{DK}|+1, ..., |\mathcal{DK}|+\frac{N-1-3K}{2}\}} |x_i|
\end{equation}
holds for {\it every} vector $\bm{x}$ in the null space of a uniform distributed $M'$-dimensional projection $A'$ in $R^{|\mathcal{DK}|+\frac{N-1-3K}{2}}$, where $M'=|\mathcal{DK}|+\frac{N-1-3K}{2}-(N-M)$.  This is because the null space of a uniform  $M$-dimensional subspaces in $R^{n}$ can be represented as $\{\bm{x}: \bm{x}=H\bm{z}, \bm{z} \in R^{N-M}\}$, where $H$ is an $N \times (N-M)$ matrix whose elements are i.i.d. Gaussian random variables $\mathcal{N}(0,1)$. With $H_l$ denoting the $l$-th row of $H$, $H_{i+1}-H_i$ is just a row vector with elements being i.i.d Gaussian random variables $\mathcal{N}(0,2)$. Noting $x_{i+1}-x_i=<H_{i+1}-H_i, \bm{z}>$, we can just think of $x_{i+1}-x_i$ as a $\sqrt{2}$ multiple of an element of a vector in a uniform $M'$-dimensional subspace in $R^{|\mathcal{DK}|+\frac{N-1-3K}{2}}$.

Now our problem reduces to determining for what values of $K$, with high probability the $RelaxedNULL$ condition \eqref{eq:relaxedcond} holds \emph{simultaneously} for \emph{every} gradient support set $\mathcal{K}$ (which determines $\mathcal{DK}$ and $\mathcal{KB}$). This falls exactly into the Grassmann angle framework \cite{Neighborlypolytope, DonohoTanner, XuHassibi} which can compute such $K$ using the Grassmann angle tools from high dimensional convex polytope theory. For details, the reader can refer to \cite{XuHassibi}, with the corresponding parameter $C$  in \cite{XuHassibi} set as $\sqrt{2}$. Figure \ref{fig:Grassmann} plots the recoverable threshold $\frac{K}{M}$ as a function of the compression ration $\frac{M}{N}$ as $N \rightarrow \infty$.

\begin{figure}
\centering
\includegraphics[width=.6\textwidth]{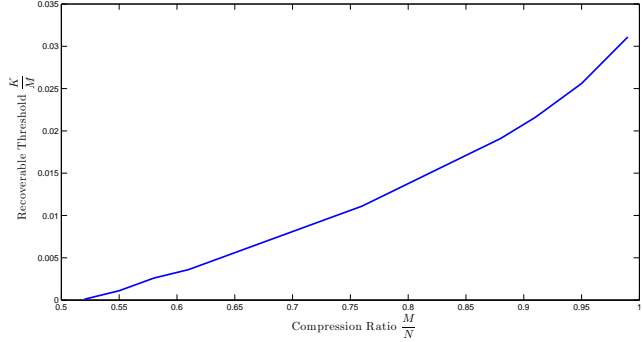}
%\end{psfrags}
\caption{Recoverable thresholds on sparsity of gradient support for TV minimization from the Grassmann angle framework \cite{XuHassibi}}\label{fig:Grassmann}
\end{figure}

\subsection{Gradient Sparsity $K$ Growing Linearly with Signal Dimension $N$}
In this part, we consider the regime of interest where the sparsity of the signal gradient grows linearly with the problem dimension. The main result is summarized in the following theorem, showing that TV minimization can allow the gradient sparsity $K$ to grow proportionally with signal dimension $N$.
\begin{theorem}
\label{thm:mainthmlinearregime}
Suppose that the measurement matrix $A$ is an $M \times N$ matrix having i.i.d. standard zero mean Gaussian elements. For any constant $0<\alpha<1$, there exists a constant $\delta>0$  such that the following statement holds true, with overwhelming probability as $M\rightarrow \infty$, $N\rightarrow \infty$, and $\frac{M}{N} \rightarrow \alpha$.

 For all subsets $\mathcal{K} \subseteq \{1,2,...,N-1\}$ with cardinality $|\mathcal{K}| \leq \delta N$, and for every nonzero vector $\bm{x}$ in the null space of $A$ (namely $A\bm{x}=0$, $\bm{x} \neq \bm{0}$),
\begin{equation}
\|(D\bm{x})_{\mathcal{K}}\|_1< \|(D\bm{x})_{\mathcal{K}^c}\|_1,
\end{equation}
where $\mathcal{K}^c=\{1,2,...,N-1\}\setminus K$.
\end{theorem}

To prove Theorem \ref{thm:mainthmlinearregime}, we first prove a uniform lower bound for the TV norm in Subsection \ref{Sec:MainNetSection}, and then utilize the lower bound to arrive at the conclusion in Subsection \ref{sec:KBound}.

\subsubsection{Uniform Lower Bound for Total Variation Norm}\label{Sec:MainNetSection}
 We consider the $(N-M)$-dimensional null space of the measurement matrix $A$. Recall that $A$ has i.i.d. standard zero mean Gaussian elements. Equivalently, a basis for the null space of $A$ can be represented by an $N \times (N-M)$ matrix $H$ with i.i.d. standard zero mean Gaussian elements. To prove the null space property for successful signal recovery using TV minimization, we only need to prove the null space property holds for those vectors $Hz$, where $z \in \mathbb{R}^{N-M}$ with $\|z\|_2=1$.

To this end, we first establish  the following claim.
\begin{theorem}
\label{thm:LDLowerBound}
With high probability as $N \rightarrow \infty$, uniformly for every $\x=Hz$ with $z \in \mathbb{R}^{N-M}$ and $\|z\|_2=1$, $$\|(D\bm{x})\|_1 \geq \gamma N,$$
where $\gamma>0$ is a sufficiently small positive constant.
\end{theorem}

We divide the proof into three parts. In the first part (Subsection \ref{sec:TVupperbound}), we establish an upper bound uniformly true for every $\x=Hz$ with $z \in \mathbb{R}^{N-M}$ and $\|z\|_2=1$. In the second part (Subsection \ref{sec:netargument}), assuming that a certain deviation bound holds true for the Total Variation norm (to be proven in Subsection \ref{sec:deviationbound}), we establish Theorem \ref{thm:LDLowerBound} using the technique of $\epsilon$-net. In the third part (Subsection \ref{sec:deviationbound}), we prove the needed deviation bound for Total Variation norm.

\paragraph{Upper Bound for the Total Variation}
\label{sec:TVupperbound}
First of all, with high probability, as $N \rightarrow \infty$, for every $\x=Hz$ with $z \in \mathbb{R}^{N-M}$ and $\|z\|_2=1$,
$$\|\x\|_2 \leq C_1 \sqrt{N},$$
where $C_1$ is a constant as $N \rightarrow \infty$. We have used the deviation bound for the largest singular value of matrices with i.i.d. Gaussian elements \cite{Geman80}.

Following this fact, we know
$$\|(D\bm{x})\|_1 \leq 2\|\x\|_1 \leq 2\sqrt{N} \|\x\|_2 \leq 2C_1 N.$$

\paragraph{ Uniform Lower Bound on Total Variation through the $\epsilon$-Net }
\label{sec:netargument}
We cover the sphere $\{z|~~\|z\|_2=1\}$ with $\epsilon$-net, where $\epsilon=C\gamma$, $C>0$ and $\gamma>0$ are constants we will choose later. $\epsilon$-net is  a finite set $V=\{v_1, ..., v_{L}\}$ on $\{z|~\|z\|_2=1\}$ such that every point $z$ from $\{z|~\|z\|_2=1\}$, there is a $v_l \in V$ such that $\|z-v_l\|_2 \leq \epsilon$. The size of the $\epsilon$-net can be taken no bigger than $(1+\frac{2}{\epsilon})^{N-M}$.

From Subsection \ref{sec:deviationbound}, we know that for every $\x=Hz$ generated by points from $\epsilon$-net,
$$\|(D\bm{x})\|_1 \geq \gamma N,$$
where $\gamma>0$ is a sufficiently small constant.

For any $z$ such that $\|z\|_2=1$, there exists a point $v_0$ (we change the subscript numbering for $V$ to index the order) in $V$ such
that $\|z-v_0\|_2\triangleq \epsilon_1 \leq \epsilon$. Let $z_1$ denote $z-v_0$,
 then $\|z_1-\epsilon_1v_1\|_2 \triangleq \epsilon_2 \leq \epsilon_1 \epsilon \leq \epsilon^2$ for
 some $v_1$ in $V$. Repeating this process, we have $z=\sum_{j\geq 0} \epsilon_j v_j$, where $\epsilon_0=1$, $\epsilon_j \leq \epsilon^j$ and $v_j \in V$.

%Thus for any $z \in R^m$, $z=\|z\|_2\sum_{j\geq 0} \epsilon_j v_j$. For any index set $K$ with $|K| \leq \beta n$,
%\begin{eqnarray*}
%\sum \limits_{i \in K} |(Hz)_i| &=& \|z\|_2 \sum \limits_{i \in K} |(\sum \limits_{j \geq 0} \epsilon_j Hv_j)_i| \\
%%& \leq & \|z\|_2 \sum \limits_{i \in T} \sum \limits_{j \geq 0} \gamma_j |(Hv_j)_i| \\
%& \leq & \|z\|_2 \sum \limits_{i \in K} \sum \limits_{j \geq 0} \gamma^{j} |(Hv_j)_i| \\
%&= & \|z\|_2  \sum \limits_{j \geq 0} \gamma^{j} \sum \limits_{i \in K} |(Hv_j)_i| \\
%%& \leq &\|z\|_2  \sum \limits_{j \geq 0} \gamma^j  (\frac{1}{2}-\delta)S\\
%& \leq & S\|z\|_2 \frac{(1-\delta)(1+\epsilon)}{(2-\delta){(1-\gamma)}}
%\end{eqnarray*}
Then
\begin{eqnarray*}
\sum \limits_{i} |(D(Hz))_i| &=& \sum \limits_{i } |D(\sum \limits_{j \geq 0} \epsilon_{j} Hv_j)_i| \\
& \geq & \sum \limits_{i}(|D(Hv_0)_i|- \sum \limits_{j \geq 1} \epsilon_{j} |D(Hv_j)_i|) \\
& \geq & \sum \limits_{i} |D(Hv_0)_i|- \sum \limits_{j \geq 1} \epsilon^{j} \sum \limits_{i}|D(Hv_j)_i| \\
& \geq &  \gamma N - \frac{\epsilon}{1-\epsilon} \times 2 C_1 N,
\end{eqnarray*}
where the first inequality follows from the triangle inequality, and the last inequality follows from the upper bound on the TV norm in Subsection \ref{sec:TVupperbound}.

We have just shown that,  for for every $\x=Hz$ with $z \in \mathbb{R}^{N-M}$ and $\|z\|_2=1$,
$$\|(D\bm{x})\|_1 \geq \gamma N-\frac{\epsilon}{1-\epsilon}\times 2 C_1 N.$$
For any arbitrary positive constant $\beta>0$, we can always take $C>0$ to be a sufficiently small constant (this does not affect the proof and conclusion in Subsection \ref{sec:deviationbound}), such that
$$\|(D\bm{x})\|_1 \geq (1-\beta)\gamma N$$
holds true for every every $\x=Hz$ with $z \in \mathbb{R}^{N-M}$ and $\|z\|_2=1$.

\paragraph{Proving the deviation bound}
\label{sec:deviationbound}

In this subsection, we prove that, for a constant $C>0$, a sufficiently small constant  $\gamma>0$, and $\epsilon=C\gamma$,
for every $\x=Hz$ generated by points from $\epsilon$-net,
$$\|(D\bm{x})\|_1 \geq \gamma N,$$
with overwhelming probability as $N \rightarrow \infty$.

We claim,  it is sufficient to prove, for a vector $\x$ with i.i.d. zero mean standard Gaussian random variables, as $N \rightarrow \infty$, with probability $e^{-N (\log(\frac{1}{\gamma})+C_2)}$ , where $C_2>0$ is a constant,
$$\|(D\bm{x})\|_1 \leq \gamma N.$$

In fact, we recall that the size of the $\epsilon$-net is at most $(1+\frac{2}{C\gamma})^{N-M}$, and notice that, for any point $z \in \mathbb{R}^{N-M}$ and $\|z\|_2=1$, the elements of $Hz$ are i.i.d. standard zero mean Gaussian random variables. By a simple union bound, with probability at most
$$P=e^{(N-M) \log(1+\frac{2}{C\gamma})}\times e^{-N (\log(\frac{1}{\gamma})+C_2)},$$
there exist some point $\x=Hz$ with $z$ from the $\epsilon$-net, such that

$$\|(D\bm{x})\|_1 \leq \gamma N.$$

No matter what $C$ we are looking at, if we take $\gamma>0$ sufficiently small, this probability $P$ converges to $0$,  as $N\ \rightarrow \infty$, $M \rightarrow \infty$, and $\frac{M}{N} \rightarrow \alpha$, where $\alpha$ is a constant. This means that with overwhelming probability, for all points from the $\epsilon$-net,
$$\|(D\bm{x})\|_1 \geq \gamma N.$$
This leads to the conclusion in Subsection \ref{sec:netargument} that, with overwhelming probability, for every $\x=Hz$ with $z \in \mathbb{R}^{N-M}$ and $\|z\|_2=1$,
$$\|(D\bm{x})\|_1 \geq (1-\beta)\gamma N.$$

Now we focus on proving the following theorem about a sequence of i.i.d. zero mean Gaussian random variables of unit variance.

\begin{theorem}
\label{thm:deviationforsinglepoint}
Suppose $x_1$, $x_2$, ..., and $x_{N}$ are $N$ independent random variables following the standard Gaussian distribution $\mathcal{N}(0,1)$. Then for all sufficiently small $\gamma>0$, the probability
$$P(\sum\limits_{i=1}^{N-1} |x_{i+1}-x_i| \leq \gamma N) \leq  e^{- (1-\mu)N (\log(\frac{1}{\gamma})+C_2+o(1))},$$
where both $\mu>0$ and $C_2$ are constants, independent of $\gamma$ and $N$; $o(1)$ goes to zero as $N \rightarrow \infty$; moreover, $\mu>0$ can be made arbitrarily small.

\end{theorem}

\begin{proof}
Let us define $t_i=|x_{i+1}-x_i|$, $1\leq i \leq N-1$. Suppose that
$$\|(D\bm{x})\|_1 \leq \gamma N,$$
then among $N$ terms in $\|(D\bm{x})\|_1$, there must be at most $\frac{N}{T} $ terms which are larger than $T\gamma$ in magnitude, where $T$ is any arbitrary positive number (say, $T=1000$); namely, there must be at least $(1-\frac{1}{T})N$ terms that are no bigger than $T\gamma$.

Let $\mathcal{M}_i$ denote the event that $t_i\geq T\gamma$, and let $\mathcal{M}_i^c$ denote the complementary event that $t_i\leq T\gamma$. We further define the indicator function $I_i$, $1\leq i \leq N$, as
\begin{equation*}
I_i=
\begin{cases}
0 &\mbox{if~} M_i\mbox{~happens,}\cr
1 & \mbox{if~} M_i^c \mbox{~happens}.
\end{cases}
\end{equation*}

And let $S$ be the set of $i$'s such that $t_i \leq T\gamma$; and $S^c=\{1,2,...,N-1\}\setminus S$ be the complement of $S$. Then the probability that $M_i$ happens for and only for $i \in S$ is
$$P(I_1, I_2,...,I_{N-1})=\prod\limits_{i=1}^{N-1} P(I_i|I_1, I_2,...,I_{i-1}).  $$
When $I_i=0$, we simply upper bound  $P(I_i|I_1, I_2,...,I_{i-1})$ by $1$; when $I_i=1$, we claim that $P(I_i|I_1, I_2,...,I_{i-1})$ is upper bounded by $\frac{1}{\sqrt{2\pi}}T\gamma$. In fact, in the Gaussian process $x_{i+1}-x_i$, no matter what values $x_1$,$x_2$, ...$x_{i-1}$ take, the probability of having a magnitude $|x_{i}-x_{i-1}|$ no larger than than $T\gamma$ is maximized when $x_{i-1}$ is equal to $0$. This leads to
$$P(I_1, I_2,...,I_{N-1})\leq (\frac{1}{\sqrt{2\pi}}T\gamma)^{|S|}.  $$
Since $|S|\geq (1-\frac{1}{T})N$, we have
$$P(I_1, I_2,...,I_{N-1})\leq (\frac{1}{\sqrt{2\pi}}T\gamma)^{(1-\frac{1}{T})N}.  $$
We have at most $\binom{N-1}{|S|}$ possibility for the set $S$ with cardinality $|S|$. So the probability
$$P_1 \triangleq P(\sum\limits_{i=1}^{N-1} |x_{i+1}-x_i| \leq \gamma N) \leq \sum\limits_{j=(1-\frac{1}{T})N}^{N-1} \binom{N-1}{j} (\frac{1}{\sqrt{2\pi}}T\gamma)^{j}.$$

From Stirling's formula and notice the $j=(1-\frac{1}{T})N$ is the biggest term in the upper bound of $P_1$, when $\gamma$ is sufficient small such that $(\frac{1}{\sqrt{2\pi}}T\gamma)$ is smaller than $1-\frac{1}{T}$,  we have
$$\log(P_1)/N \rightarrow \left(H(\frac{1}{T}) +(1-\frac{1}{T}) \log(\frac{1}{\sqrt{2\pi}} T\gamma) \right)+o(1),$$
where $H(p)=p \log(\frac{1}{p})+(1-p) \log(\frac{1}{1-p})$ is the entropy function, and $o(1)$ is a term that goes to $0$ as $N \rightarrow \infty$.

Since we can pick an arbitrarily big constant $T$, we have the theorem statement by simply taking $\mu=\frac{1}{T}$.
\end{proof}

We remark that Theorem \ref{thm:deviationforsinglepoint} eventually leads to the proof of Theorem \ref{thm:LDLowerBound}.

\subsubsection{Upper Bound on the Partial Total Variation Norm}
\label{sec:KBound}
In this section, we prove the following theorem.

\begin{theorem}
\label{thm:partialTVnorm}
Suppose a matrix $H$ is an $N \times (N-M)$ matrix having i.i.d. standard zero mean Gaussian elements. For any constant $0<\alpha<1$ and any positive constant $\gamma>0$, there exists a constant $\delta>0$  such that the following statement holds true, with overwhelming probability as $M\rightarrow \infty$, $N\rightarrow \infty$, and $\frac{M}{N} \rightarrow \alpha$.

For all subsets $\mathcal{K} \subseteq \{1,2,...,N-1\}$ with cardinality $|\mathcal{K}| \leq \delta N$, and for every $\x=Hz$ with $z \in \mathbb{R}^{N-M}$ and $\|z\|_2=1$,
\begin{equation}
\|(D\bm{x})_{\mathcal{K}}\|_1< \frac{1}{2} \gamma N.
\end{equation}
\end{theorem}

We notice that $\|(D\bm{x})_{\mathcal{K}}\|_1 \leq 2 \sum\limits_{j \in \mathcal{M}} |x_j|$, where $\mathcal{M}$ is the set of indices $j$'s that $x_j$ is involved in the expression $\|(D\bm{x})_{\mathcal{K}}\|_1$. Because of this, and the fact that the cardinality $|\mathcal{M}| \leq 2 |\mathcal{K}|$, we can use the same methodology in \cite{L_pMinimization} to prove Theorem \ref{thm:partialTVnorm}, based on the uniform lower bound result we have from Theorem \ref{thm:LDLowerBound}.

\subsection{Stability of Total Variation Minimization in Signal Recovery}
Our results also show that TV minimization provides stable signal recovery when the signal does not have exactly $K$-sparse gradient and there are noise contained in the measurements. In particular, we assume that the noise level is $\|A\bm{x}-\bm{y}\|_2\leq\epsilon$. Then, we solve
\begin{equation}\label{eq:TVminNoise}
\min_{\bm{x}}~\|D\bm{x}\|_1~~\mbox{s.t.}~\|A\bm{x}-\bm{y}\|_2\leq\epsilon
\end{equation}
Our results show that \eqref{eq:TVminNoise} provides robust signal recovery. The robustness result is summarized in the following theorem.

\begin{theorem}
\label{thm:robustnessstatement}
Suppose that the measurement matrix $A$ is an $M \times N$ matrix having i.i.d. standard zero mean Gaussian elements. For any constant $0<\alpha<1$ and any constant $0<C<1$, there exists a constant $\delta>0$ and a constant $\beta>0$  such that the following statement holds true, with overwhelming probability as $M\rightarrow \infty$, $N\rightarrow \infty$, and $\frac{M}{N} \rightarrow \alpha$.

For all $\bm{x} \in \mathbb{R}^{N}$ and $\|A\bm{x}-\bm{y}\|_2\leq\epsilon$,
%and for all subsets $\mathcal{K} \subseteq \{1,2,...,N-1\}$ with cardinality $|\mathcal{K}| \leq \delta N$,
the solution $\hat{\bm{x}}$ to \eqref{eq:TVminNoise} satisfies,
%for all subsets $\mathcal{K} \subseteq \{1,2,...,N-1\}$ with cardinality $|\mathcal{K}|:=K \leq \delta N$,
\begin{equation}
\|\bm{x}-\hat{\bm{x}}\|_2 \leq \left(\frac{2(1+C)}{\beta(1-C)}\right)\frac{\min_{|\mathcal{K}|\leq \delta N}\|(D\bm{x})_{\mathcal{K}^c}\|_1}{\sqrt{N}}+\left(2+\frac{4(1+C)}{\beta(1-C)}\right)\frac{\epsilon}{\sqrt{N}-\sqrt{M}},
\end{equation}
\end{theorem}

Note that we do not normalize the sensing matrix $A$. Therefore, it is not surprising that the error is of order $\epsilon/\sqrt{N}$, since $\epsilon$ grows with $N$ if the relative error $\|A\bm{x}-\bm{y}\|_2/\|A\bm{x}\|_2$ is fixed. To prove the theorem, we need the following theorems that provides the so-called ``Balanced Condition for TV Norm'' Theorem \ref{thm:stableBALcondition}, the ``Almost Euclidean Property ''.

%\begin{theorem}
%\label{thm:stabilitystatement}
%Suppose that the measurement matrix $A$ is an $M \times N$ matrix having i.i.d. standard zero mean Gaussian elements. For any constant $0<\alpha<1$ and any constant $C>1$, there exists a constant $\delta>0$ and a constant $\beta>0$  such that the following statement holds true, with overwhelming probability as $M\rightarrow \infty$, $N\rightarrow \infty$, and $\frac{M}{N} \rightarrow \alpha$.
%
%
%For all $\bm{x} \in \mathbb{R}^{N}$ and $\bm{y}=A\bm{x}$, the solution $\hat{\bm{x}}$ to the TV minimization problem (\ref{eq:TVprogram}) satisfies, for all subsets $\mathcal{K} \subseteq \{1,2,...,N-1\}$ with cardinality $|\mathcal{K}| \leq \delta N$,
%\begin{equation}
%\|D(\bm{x}-\hat{\bm{x}})\|_1 \leq \frac{2(C+1)}{C-1} \|(D\bm{x})_{\mathcal{K}^c}\|_1,
%\end{equation}
%and
%\begin{equation}
%\|\bm{x}-\hat{\bm{x}}\|_2 \leq \frac{2\beta (C+1)}{C-1} \frac{\|(D\bm{x})_{\mathcal{K}^c}\|_1}{\sqrt{N}},
%\end{equation}
%where $\mathcal{K}^c=\{1,2,...,N-1\}\setminus K$.
%
%\end{theorem}

%\begin{proof}
%This theorem statement follows from the ``Balanced Condition for TV Norm'' Theorem \ref{thm:stableBALcondition}, the ``Almost Euclidean Property '' Theorem  \ref{thm:almostEuclideanTV} and similar arguments for $\ell_1$ minimization in \cite{XuHassibi}.
%\end{proof}

\begin{theorem} [Balanced Condition for TV Norm]
\label{thm:stableBALcondition}
Suppose that the measurement matrix $A$ is an $M \times N$ matrix having i.i.d. standard zero mean Gaussian elements. For any constant $0<\alpha<1$ and any constant $0<C<1$, there exists a constant $\delta>0$  such that the following statement holds true, with overwhelming probability as $M\rightarrow \infty$, $N\rightarrow \infty$, and $\frac{M}{N} \rightarrow \alpha$.

For all subsets $\mathcal{K} \subseteq \{1,2,...,N-1\}$ with cardinality $|\mathcal{K}| \leq \delta N$, and for every nonzero vector $\bm{x}$ in the null space of $A$ (namely $A\bm{x}=0$, $\bm{x} \neq \bm{0}$),

\begin{equation}
\|(D\bm{x})_{\mathcal{K}}\|_1< C \|(D\bm{x})_{\mathcal{K}^c}\|_1,
\end{equation}
where $\mathcal{K}^c=\{1,2,...,N-1\}\setminus K$;

\end{theorem}

\begin{proof}
The proof is essentially similar to the proof of Theorem \ref{thm:mainthmlinearregime}, except that in the proof of Theorem \ref{thm:mainthmlinearregime}, we take $C=1$.
\end{proof}

We also have the following almost Euclidean property for the TV norm.

\begin{theorem}[Almost Euclidean Property for TV Norm]
\label{thm:almostEuclideanTV}
Suppose that the measurement matrix $A$ is an $M \times N$ matrix having i.i.d. standard zero mean Gaussian elements. For any constant $0<\alpha<1$, there exists a constant $\beta>0$  such that the following statement holds true, with overwhelming probability as $M\rightarrow \infty$, $N\rightarrow \infty$, and $\frac{M}{N} \rightarrow \alpha$.

For every nonzero vector $\bm{x}$ in the null space of $A$ (namely $A\bm{x}=0$, $\bm{x} \neq \bm{0}$),
\begin{equation}
\|D(\bm{x})\|_1 \geq \beta \sqrt{N} \|\bm{x}\|_2.
\end{equation}

\end{theorem}

\begin{proof}
This follows from the lower bound in Theorem \ref{thm:LDLowerBound} and the upper bound results on $\ell_2$ norm in Subsection \ref{sec:TVupperbound}.
\end{proof}

Now we have the proof of Theorem \ref{thm:robustnessstatement}.

\begin{proof}[Proof of Theorem \ref{thm:robustnessstatement}]
Let $\mathcal{K} \subseteq \{1,2,...,N-1\}$ be a minimizer of $\min_{|\mathcal{K}|\leq \delta N}\|(D\bm{x})_{\mathcal{K}^c}\|_1$.
Let $\bm{w}=\bm{x}-\hat{\bm{x}}$, and we decompose it orthogonally as $\bm{w}=\bm{w}_1+\bm{w}_2$, where $\bm{w}_1$ and $\bm{w}_2$ are in the null space of $A$ and the range of $A^T$ respectively. Then, we have
\begin{equation}\label{eq:ProofRobustNoise3}
\|\bm{x}-\hat{\bm{x}}\|_2\leq \|\bm{w}_1\|_2+\|\bm{w}_2\|_2.
\end{equation}
Since $\bm{w}_2$ is in the range of $A^T$,
\begin{equation}\label{eq:ProofRobustNoise4}
\|\bm{w}_2\|_2\leq \frac{1}{\sigma_{min}(A)}\|A\bm{w}_2\|_2
=\frac{1}{\sigma_{min}(A)}\|A\bm{w}\|_2
\leq \frac{1}{\sigma_{min}(A)}(\|A\bm{x}-\bm{y}\|_2+\|A\hat{\bm{x}}-\bm{y}\|_2)
\leq \frac{2}{\sigma_{min}(A)}\epsilon.
\end{equation}
Since $\bm{w}_1$ is in the kernel of $A$, by Theorem \ref{thm:almostEuclideanTV}, we have
\begin{equation}\label{eq:ProofRobustNoise5}
\|\bm{w}_1\|_2\leq \frac{1}{\beta\sqrt{N}}\|D\bm{w}_1\|_1
\end{equation}
Let us estimate $\|D\bm{w}_1\|_1$. The minimality of $\|D\hat{\bm{x}}\|_1$ implies
\begin{equation*}
\begin{split}
\|(D\bm{x})_{\mathcal{K}^c}\|_1+\|(D\bm{x})_\mathcal{K}\|_1&=\|D\bm{x}\|_1\geq \|D\hat{\bm{x}}\|_1 = \|D\bm{x}+D\bm{w}\|_1
= \|(D\bm{x}+D\bm{w})_{\mathcal{K}}\|_1+\|(D\bm{x}+D\bm{w})_{\mathcal{K}^c}\|_1\cr
&\geq \|(D\bm{x})_{\mathcal{K}}\|_1-\|(D\bm{w})_{\mathcal{K}}\|_1+\|(D\bm{w})_{\mathcal{K}^c}\|_1-\|(D\bm{x})_{\mathcal{K}^c}\|_1,
\end{split}
\end{equation*}
which leads to
$$
2\|(D\bm{x})_{\mathcal{K}^c}\|_1+\|(D\bm{w})_{\mathcal{K}}\|_1\geq\|(D\bm{w})_{\mathcal{K}^c}\|_1.
$$
Therefore,
\begin{equation*}
\begin{split}
2\|(D\bm{x})_{\mathcal{K}^c}\|_1+&\|(D\bm{w}_1)_\mathcal{K}\|_1+\|(D\bm{w}_2)_\mathcal{K}\|_1\cr
&\geq2\|(D\bm{x})_{\mathcal{K}^c}\|_1+\|(D\bm{w})_{\mathcal{K}}\|_1\geq
\|(D\bm{w})_{\mathcal{K}^c}\|_1\geq \|(D\bm{w}_1)_{\mathcal{K}^c}\|_1-\|(D\bm{w}_2)_{\mathcal{K}^c}\|_1,
\end{split}
\end{equation*}
and thus
\begin{equation}\label{eq:ProofRobustNoise1}
\begin{split}
\|(D\bm{w}_1)_{\mathcal{K}^c}\|_1
&\leq 2\|(D\bm{x})_{\mathcal{K}^c}\|_1+\|D\bm{w}_2\|_1+\|(D\bm{w}_1)_\mathcal{K}\|_1
\leq 2\|(D\bm{x})_{\mathcal{K}^c}\|_1+2\|\bm{w}_2\|_1+\|(D\bm{w}_1)_\mathcal{K}\|_1\cr
&\leq 2\|(D\bm{x})_{\mathcal{K}^c}\|_1+\frac{4}{\sigma_{\min}(A)}\sqrt{N}\epsilon+\|(D\bm{w}_1)_\mathcal{K}\|_1.
\end{split}
\end{equation}
Moreover, by Theorem \ref{thm:stableBALcondition}
$$
\|(D\bm{w}_1)_{\mathcal{K}}\|_1\leq C \|(D\bm{w}_1)_{\mathcal{K}^c}\|_1.
$$
This together with \eqref{eq:ProofRobustNoise1} implies
$$
\|(D\bm{w}_1)_{\mathcal{K}}\|_1\leq 2C\|(D\bm{x})_{\mathcal{K}^c}\|_1+\frac{4C}{\sigma_{\min}(A)}\sqrt{N}\epsilon+C\|(D\bm{w}_1)_\mathcal{K}\|_1
$$
and further
$$
\|(D\bm{w}_1)_{\mathcal{K}}\|_1\leq \frac{2C}{1-C}\|(D\bm{x})_{\mathcal{K}^c}\|_1+ \frac{4C}{(1-C)\sigma_{\min}(A)}\sqrt{N}\epsilon.
$$
Substituting it into \eqref{eq:ProofRobustNoise1} again yields
\begin{equation*}
\begin{split}
\|(D\bm{w}_1)_{\mathcal{K}^c}\|_1&\leq (2+\frac{2C}{1-C})\|(D\bm{x})_{\mathcal{K}^c}\|_1
\left(\frac{4}{\sigma_{\min}(A)}+\frac{4C}{(1-C)\sigma_{\min}(A)}\right)\sqrt{N}\epsilon\cr
&=\frac{2}{1-C}\|(D\bm{x})_{\mathcal{K}^c}\|_1+\frac{4}{(1-C)\sigma_{\min}(A)}\sqrt{N}\epsilon.
\end{split}
\end{equation*}
We obtain
\begin{equation}\label{eq:ProofRobustNoise2}
\|D\bm{w}_1\|_1=\|(D\bm{w}_1)_{\mathcal{K}}\|_1+\|(D\bm{w}_1)_{\mathcal{K}^c}\|_1
\leq \frac{2(1+C)}{(1-C)}\|(D\bm{x})_{\mathcal{K}^c}\|_1+\frac{4(1+C)}{(1-C)\sigma_{\min}(A)}\sqrt{N}\epsilon.
\end{equation}
Finally, combine \eqref{eq:ProofRobustNoise3}, \eqref{eq:ProofRobustNoise4}, \eqref{eq:ProofRobustNoise5}, and \eqref{eq:ProofRobustNoise2} and get
$$
\|\bm{x}-\hat{\bm{x}}\|_2\leq \left(\frac{2(1+C)}{\beta(1-C)}\right)\frac{\|(D\bm{x})_{\mathcal{K}^c}\|_1}{\sqrt{N}}+ \left(2+\frac{4(1+C)}{\beta(1-C)}\right)\frac{\epsilon}{\sigma_{\min}(A)}.
$$
To conclude the proof, we use the well-known fact that the minimum non-zero singular values of an $M\times N$ Gaussian random matrix is the order of $\sqrt{N}-\sqrt{M}$.
\end{proof}

\section{Extension to Multidimensional signals}
\label{sec:multidimensional}
In this section, we extend our results to $d$-dimensional ($d\geq 2$, for example $d=2$ for image and $d=3$ for videos) signal vectors. We get results that are comparable to those in \cite{TV2d, TVMulti}. In particular, let $\bm{X}\in\mathbb{R}^{N^d}$ be a multi-indexed vector that is from a $d$-dimensional signal. Let $A\in\mathbb{R}^{M\times N^d}$ be a measurement matrix whose elements are i.i.d. Gaussian random variables, and $\bm{Y}=A\bm{X}$ be its corresponding measurements of $\bm{X}$. Define $D\bm{X}$ be the discrete gradient of $\bm{X}$. Assume that $D\bm{X}$ contains at most $K$ nonzero entries. In order to recover $\bm{X}$, similar to \eqref{eq:TVprogram}, we solve the following minimization
\begin{equation}\label{eq:TVprogramMulti}
\min_{\bm{X}}~\|D\bm{X}\|_1,\qquad\mbox{subject~to}\quad \bm{Y}=A\bm{X}.
\end{equation}
In the remaining of this section, we prove that the unique solution of \eqref{eq:TVprogramMulti} is exactly the original $\bm{X}$ with high probability, as long as
\begin{equation*}
M\geq
\begin{cases}
C_1 K\log_2^2N\ln N &\mbox{if~} d=2\cr
C_2 K\ln N & \mbox{if~} d>2.
\end{cases}
\end{equation*}
%as $K \rightarrow \infty$ and $N \rightarrow \infty$,
where $C_1>0$ and $C_2>0$ are two constants depending on $d$. Note that $\|D\bm{X}\|_1$ in \eqref{eq:TVprogramMulti} is the anisotropic TV. Our proof can be generalized to isotropic TV without too much difficulty.

Similar to Theorem \ref{thm:TVnullspacecondition1}, a sufficient condition for the original $\bm{X}$ being the unique solution of \eqref{eq:TVprogramMulti} is the null space condition \eqref{eq:TVnullspacecondition}. Different from $1$-dimensional case, this null space condition is only a sufficient condition for higher dimensional signals. Then, using the escape through the mesh theorem, this null space condition holds true with high probability if the Gaussian width satisfies $w(\mathcal{S}_d)<\sqrt{M}-\frac{1}{2\sqrt{M}}$, where
$$
\mathcal{S}_d=\{\bm{X}\in\mathbb{R}^{N^d}~:~\|\bm{X}\|_2= 1,\quad\mbox{and}\quad \|(D\bm{X})_{\mathcal{K}}\|_1\geq \|(D\bm{X})_{\mathcal{K}^c}\|_1 ~\exists \mathcal{K}\subset\{1,\ldots,N\}^{d}\times\{1,\ldots,d\}\mbox{~s.t.~}|\mathcal{K}|\leq K\}.
$$
Given any vector $\bm{X}\in\mathcal{S}_d$, we have
\begin{equation*}
\begin{split}
\|D\bm{X}\|_1&=\|(D\bm{X})_{\mathcal{K}^c}\|_1+\|(D\bm{X})_{\mathcal{K}}\|_1\leq 2\|(D\bm{X})_{\mathcal{K}}\|_1\leq 2\sqrt{K}\|(D\bm{X})_{\mathcal{K}}\|_2\cr
&\leq 2\sqrt{K}\|D\bm{X}\|_2\leq 4\sqrt{d}\sqrt{K}\|\bm{X}\|_2\leq4\sqrt{d}\sqrt{K}.
\end{split}
\end{equation*}
We have used the fact that $\|D\bm{X}\|_2\leq 2\sqrt{d} \|\bm{X}\|_2$. Therefore,
$$
\mathcal{S}_d\subset\widetilde{\mathcal{S}}_d:=\{\bm{X}\in\mathbb{R}^{N^d}~:~\|\bm{X}\|_2\leq 1,~~\|D\bm{X}\|_1\leq 4\sqrt{d}\sqrt{K}\}.
$$

In the following, we estimate the Gaussian width of $\widetilde{\mathcal{S}}_d$. Similar to $1$-dimensional signal, we consider only the case where $N=2^L$. For any $\bm{X}\in\widetilde{\mathcal{S}}_d$, we decompose $\bm{X}$ according to Haar wavelet transform for $d$-dimensional vector as
\begin{equation}\label{Eq:Haarnd}
\bm{X} =\sum_{\ell=1}^{L}\sum_{\bm{i}\in\{0,1\}^{d}\setminus\bm{0}}\hat{\bm{Z}}^{(\ell,\bm{i})}+\hat{\bm{Y}}^{(L)},
\end{equation}
where
$$
\hat{\bm{Z}}^{(\ell,\bm{i})}=\bm{Z}^{(\ell,\bm{i})}\otimes\bm{H}^{(\bm{i})}\otimes\bm{1}_{2^{\ell-1}},\quad
\bm{Z}^{(\ell,\bm{i})}\in\mathbb{R}^{(N/{2^\ell})^d}
$$
and
$$
\hat{\bm{Y}}^{(L)}=\bm{Y}^{(L)}\otimes\bm{1}_{N},\quad \bm{Y}^{(L)}\in\mathbb{R}.
$$
Here $\bm{1}_n\in\mathbb{R}^{n^d}$ is the $d$-dimensional vector whose entries are all $1$, and
$\otimes$ is the Kronecker product, i.e., $\bm{A}\otimes\bm{B}$ is the block $d$-dimensional matrix whose $(j_1,j_2,\ldots,j_d)$ block is $A_{j_1j_2\ldots j_d}\bm{B}$.
Moreover, $\bm{H}^{(\bm{i})}\in \mathbb{R}^{2^d}$ with $\bm{i}=(i_1,i_2,\ldots,i_d)$ is the (scaled) Haar filter defined by
$$
H^{(\bm{i})}_{j_1j_2\ldots j_d}=\prod_{k=1}^{d}h^{(i_k)}_{j_k},
\quad\mbox{with}\quad \bm{h}^{(0)}=[1~1],~\bm{h}^{(1)}=[1~-1].
$$
In particular, we have $\bm{H}^{(0)}=\bm{1}_2$.

The decomposition \eqref{Eq:Haarnd} is done recursively as follows. We first decompose $\bm{X}:=\hat{\bm{Y}}^{(0)}=\hat{\bm{Y}}^{(1)}+\sum_{\bm{i}\in\{0,1\}^{d}\setminus\bm{0}}\hat{\bm{Z}}^{(1,\bm{i})}$, where
$$
\hat{\bm{Y}}^{(1)} = \bm{Y}^{(1)}\otimes \bm{1}_2,
%[y_1^{(1)}~y_1^{(1)}~y_2^{(1)}~y_2^{(1)}~\ldots~\ldots~y_{N/2}^{(1)}~y_{N/2}^{(1)}],
\qquad Y_{\bm{k}}^{(1)} = \frac{\sum_{\bm{j}\in\{0,1\}^{d}}H_{\bm{j}}^{(\bm{0})}X_{2\bm{k}-\bm{j}}}{2^d}=
\frac{\sum_{\bm{j}\in\{0,1\}^{d}}X_{2\bm{k}-\bm{j}}}{2^d},
$$
and
$$
\hat{\bm{Z}}^{(1,\bm{i})} = \bm{Z}^{(1,\bm{i})}\otimes\bm{H}^{(\bm{i})},\qquad Z_{\bm{k}}^{(1,\bm{i})} = \frac{\sum_{\bm{j}\in\{0,1\}^{d}}H^{(\bm{i})}_{\bm{j}}X_{2\bm{k}-\bm{j}}}{2^d}.
$$
%Since the matrix $\bm{H}_1=[\bm{h}^{(0)}~\bm{h}^{(1)}]$ satisfies $\bm{H}_1^T\bm{H}_1=2\bm{I}_2$, the matrix $\bm{H}_d:=[\bm{H}^{(i)}]_{\bm{i}\in\{0,1\}^{d}}$ also satisfies $\bm{H}_d^T\bm{H}_d=2^d\bm{I}_{2^d}$, which ensures the decomposition is valid.
One can check that $\hat{\bm{Y}}^{(0)}=\hat{\bm{Y}}^{(1)}+\sum_{\bm{i}\in\{0,1\}^{d}\setminus\bm{0}}\hat{\bm{Z}}^{(1,\bm{i})}$. Furthermore, it can be easily shown that this decomposition is an orthogonal decomposition.
Then, we further decompose
\begin{equation}\label{Eq:HaarDdDL2}
\hat{\bm{Y}}^{(1)}=\hat{\bm{Y}}^{(2)}+\sum_{\bm{i}\in\{0,1\}^{d}\setminus\bm{0}}\hat{\bm{Z}}^{(2,\bm{i})}
\end{equation}
where
$$
\hat{\bm{Y}}^{(2)} = \bm{Y}^{(2)}\otimes \bm{1}_4,
%[y_1^{(1)}~y_1^{(1)}~y_2^{(1)}~y_2^{(1)}~\ldots~\ldots~y_{N/2}^{(1)}~y_{N/2}^{(1)}],
\qquad Y_{\bm{k}}^{(2)} = \frac{\sum_{\bm{j}\in\{0,1\}^{d}}H_{\bm{j}}^{(\bm{i})}Y^{(1)}_{2\bm{k}-\bm{j}}}{2^d}
=\frac{\sum_{\bm{j}\in\{0,1\}^{d}}Y^{(1)}_{2\bm{k}-\bm{j}}}{2^d},
$$
and
$$
\hat{\bm{Z}}^{(2,\bm{i})} = \bm{Z}^{(2,\bm{i})}\otimes\bm{H}^{(\bm{i})}\otimes \bm{1}_2,\qquad Z_{\bm{k}}^{(2,\bm{i})} = \frac{\sum_{\bm{j}\in\{0,1\}^{d}}H^{(\bm{i})}_{\bm{j}}Y^{(1)}_{2\bm{k}-\bm{j}}}{2^d}.
$$
Again, one can check \eqref{Eq:HaarDdDL2} holds true and is an orthogonal decomposition. Generally, at level $\ell$, we have that
$$
\hat{\bm{Y}}^{(\ell)}=\bm{Y}^{(\ell)}\otimes \bm{1}_{2^{\ell}},\qquad
\bm{Y}^{(\ell)}\in\mathbb{R}^{(N/2^{\ell})^d},
$$
and we decompose it as
\begin{equation}\label{Eq:HaarDdD}
\hat{\bm{Y}}^{(\ell)}=\hat{\bm{Y}}^{(\ell+1)}+\sum_{\bm{i}\in\{0,1\}^{d}\setminus\bm{0}}\hat{\bm{Z}}^{(\ell+1,\bm{i})},
\end{equation}
where
$$
\hat{\bm{Y}}^{(\ell+1)} = \bm{Y}^{(\ell+1)}\otimes \bm{1}_{2^{\ell+1}},
%[y_1^{(1)}~y_1^{(1)}~y_2^{(1)}~y_2^{(1)}~\ldots~\ldots~y_{N/2}^{(1)}~y_{N/2}^{(1)}],
\qquad Y_{\bm{k}}^{(\ell+1)} = \frac{\sum_{\bm{j}\in\{0,1\}^{d}}H_{\bm{j}}^{(\bm{0})}Y^{(\ell)}_{2\bm{k}-\bm{j}}}{2^d}
=\frac{\sum_{\bm{j}\in\{0,1\}^{d}}Y^{(\ell)}_{2\bm{k}-\bm{j}}}{2^d},
$$
and
$$
\hat{\bm{Z}}^{(\ell+1,\bm{i})} = \bm{Z}^{(\ell+1,\bm{i})}\otimes\bm{H}^{(\bm{i})}\otimes \bm{1}_{2^{\ell}},\qquad Z_{\bm{k}}^{(\ell+1,\bm{i})} = \frac{\sum_{\bm{j}\in\{0,1\}^{d}}H^{(\bm{i})}_{\bm{j}}Y^{(\ell)}_{2\bm{k}-\bm{j}}}{2^d}.
$$

The decomposition \eqref{Eq:Haarnd} has the following properties.
\begin{itemize}
\item Obviously, components in decomposition \eqref{Eq:Haarnd} are orthogonal to each others. Consequently,
$$
\|\bm{X}\|_2^2 = \sum_{\ell=1}^L\sum_{\bm{i}\in\{0,1\}^d\setminus\bm{0}}\|\hat{\bm{Z}}^{(\ell,\bm{i})}\|_2^2+\|\hat{\bm{Y}}^{(L)}\|_2^2 =
\sum_{\ell=1}^L\left(2^{d\ell}\sum_{\bm{i}\in\{0,1\}^d\setminus\bm{0}}\|\hat{\bm{Z}}^{(\ell,\bm{i})}\|_2^2\right)+2^{dL}\|\hat{\bm{Y}}^{(L)}\|_2^2.
$$
Since $\bm{X}\in\tilde{\mathcal{S}}$ implies $\|\bm{X}\|_2^2\leq 1$, we have
\begin{equation}\label{Eq:Prop1nd}
\sum_{\ell=1}^L\left(2^{d\ell}\sum_{\bm{i}\in\{0,1\}^d\setminus\bm{0}}\|\hat{\bm{Z}}^{(\ell,\bm{i})}\|_2^2\right)+2^{dL}\|\hat{\bm{Y}}^{(L)}\|_2^2
\leq 1
\end{equation}
\item It can be shown that
\begin{equation}\label{Eq:Prop20dD2}
\|D\hat{\bm{Y}}^{(\ell)}\|_1\leq \|D\hat{\bm{Y}}^{(\ell-1)}\|_1
\end{equation}
and, consequently,
\begin{equation}\label{Eq:Prop20dD}
\|D\hat{\bm{Y}}^{(\ell)}\|_{(K)}^*\leq \|D\bm{X}\|_{(K)}^*\leq 4\sqrt{d}\sqrt{K}.
\end{equation}
Let $D_i$ be the difference matrix along the $i$-th dimension. Then, similar to the 1-D case, one can show that
$\|D_i\bm{Y}^{(\ell)}\|_1\leq 2^{d-1}\cdot\frac{2}{2^d} \|D_i\bm{Y}^{(\ell-1)}\|_1=\|D_i\bm{Y}^{(\ell-1)}\|_1$. Summing over $i$ yields \eqref{Eq:Prop20dD2}.

\item Furthermore, for any vector $\bm{G}$, we have
$$
\sum_{\bm{i}\in\{0,1\}^d\setminus\bm{0}}\langle\bm{G},\hat{\bm{Z}}^{(\ell,\bm{i})}\rangle\leq \frac{2^d(2^d-1)}{2^{d-1}}\frac{4\sqrt{d}\sqrt{K}}{2^{(\ell-1)(d-1)}}\|\tilde{\bm{G}}^{(\ell-1)}\|_{\infty}
=8\sqrt{d}(2^d-1)\frac{\sqrt{K}}{2^{(\ell-1)(d-1)}}\|\tilde{\bm{G}}^{(\ell-1)}\|_{\infty},
$$
where $\tilde{\bm{G}}^{(\ell-1)}\in\mathbb{R}^{(N/2^{\ell-1})^d}$ is a $d$-dimensional signal whose $\bm{i}$-th entry is the sum of the entries of $\bm{G}$ on the $\bm{i}$-th block of size $2^{\ell-1}\times2^{\ell-1}$.
For simplicity, we prove it for $d=2$. The remaining case can be shown analogously. When $d=2$, we have the four filters are
$$
\bm{H}^{(0,0)}=\left[\begin{matrix}1&1\cr 1&1\end{matrix}\right],\quad
\bm{H}^{(1,0)}=\left[\begin{matrix}1&-1\cr 1&-1\end{matrix}\right],\quad
\bm{H}^{(0,1)}=\left[\begin{matrix}1&1\cr -1&-1\end{matrix}\right],\quad
\bm{H}^{(1,1)}=\left[\begin{matrix}1&-1\cr -1&1\end{matrix}\right].
$$
Let $D_1$ and $D_2$ be finite difference along the horizontal and vertical direction respectively. Then it is easy to check that
\begin{equation*}
\begin{split}
&\|D(a_{00}\bm{H}^{(0,0)}\otimes\bm{1}_{2^{\ell-1}}+a_{10}\bm{H}^{(1,0)}\otimes\bm{1}_{2^{\ell-1}}+
a_{01}\bm{H}^{(0,1)}\otimes\bm{1}_{2^{\ell-1}}+a_{11}\bm{H}^{(1,1)}\otimes\bm{1}_{2^{\ell-1}})\|_1\cr
&=\|D_1(a_{00}\bm{H}^{(0,0)}\otimes\bm{1}_{2^{\ell-1}}+a_{10}\bm{H}^{(1,0)}\otimes\bm{1}_{2^{\ell-1}}+
a_{01}\bm{H}^{(0,1)}\otimes\bm{1}_{2^{\ell-1}}+a_{11}\bm{H}^{(1,1)}\otimes\bm{1}_{2^{\ell-1}})\|_1\cr
&\qquad+\|D_2(a_{00}\bm{H}^{(0,0)}\otimes\bm{1}_{2^{\ell-1}}+a_{10}\bm{H}^{(1,0)}\otimes\bm{1}_{2^{\ell-1}}+
a_{01}\bm{H}^{(0,1)}\otimes\bm{1}_{2^{\ell-1}}+a_{11}\bm{H}^{(1,1)}\otimes\bm{1}_{2^{\ell-1}})\|_1\cr
&=2^{\ell-1}(|a_{01}+a_{11}|+|a_{01}-a_{11}|+|a_{10}+a_{11}|+|a_{10}-a_{11}|).
\end{split}
\end{equation*}
Therefore,
\begin{equation*}
\begin{split}
2^{\ell-1}(\|\bm{Z}^{(\ell,01)}+&\bm{Z}^{(\ell,11)}\|_1+\|\bm{Z}^{(\ell,01)}-\bm{Z}^{(\ell,11)}\|_1
+\|\bm{Z}^{(\ell,10)}+\bm{Z}^{(\ell,11)}\|_1+\|\bm{Z}^{(\ell,10)}-\bm{Z}^{(\ell,11)}\|_1)\cr
&\leq \|D(\hat{\bm{Y}}^{(\ell)}+\hat{\bm{Z}}^{(\ell,10)}+\hat{\bm{Z}}^{(\ell,01)}+\hat{\bm{Z}}^{(\ell,11)})\|_1
=\|D\hat{\bm{Y}}^{(\ell-1)}\|_1\leq 4\sqrt{2}\sqrt{K}.
\end{split}
\end{equation*}
Furthermore, if we let $\tilde{\bm{G}}_{oo}^{(\ell-1)}$ be a down sample of $\tilde{\bm{G}}^{(\ell-1)}$ on odd-odd indices and similarly $\tilde{\bm{G}}_{oe}^{(\ell-1)}$, $\tilde{\bm{G}}_{eo}^{(\ell-1)}$, and $\tilde{\bm{G}}_{ee}^{(\ell-1)}$, then
\begin{equation*}
\begin{split}
&\langle\bm{G},\hat{\bm{Z}}^{(\ell,10)}+\hat{\bm{Z}}^{(\ell,01)}+\hat{\bm{Z}}^{(\ell,11)}\rangle
=\langle\tilde{\bm{G}}^{(\ell-1)},{\bm{Z}}^{(\ell,10)}\otimes\bm{H}^{(1,0)}+
{\bm{Z}}^{(\ell,01)}\otimes\bm{H}^{(0,1)}+{\bm{Z}}^{(\ell,11)}\otimes\bm{H}^{(1,1)}\rangle\cr
=&\langle\tilde{\bm{G}}_{oo}^{(\ell-1)},{\bm{Z}}^{(\ell,01)}+{\bm{Z}}^{(\ell,10)}+{\bm{Z}}^{(\ell,11)}\rangle
+\langle\tilde{\bm{G}}_{oe}^{(\ell-1)},-{\bm{Z}}^{(\ell,01)}+{\bm{Z}}^{(\ell,10)}-{\bm{Z}}^{(\ell,11)}\rangle\cr
&+\langle\tilde{\bm{G}}_{eo}^{(\ell-1)},{\bm{Z}}^{(\ell,01)}-{\bm{Z}}^{(\ell,10)}-{\bm{Z}}^{(\ell,11)}\rangle
+\langle\tilde{\bm{G}}_{ee}^{(\ell-1)},-{\bm{Z}}^{(\ell,01)}-{\bm{Z}}^{(\ell,10)}+{\bm{Z}}^{(\ell,11)}\rangle\cr
\leq&\|\tilde{\bm{G}}_{oo}^{(\ell-1)}\|_{\infty}\|{\bm{Z}}^{(\ell,01)}+{\bm{Z}}^{(\ell,10)}+{\bm{Z}}^{(\ell,11)}\|_1
+\|\tilde{\bm{G}}_{oe}^{(\ell-1)}\|_{\infty}\|-{\bm{Z}}^{(\ell,01)}+{\bm{Z}}^{(\ell,10)}-{\bm{Z}}^{(\ell,11)}\|_1\cr
&+\|\tilde{\bm{G}}_{eo}^{(\ell-1)}\|_{\infty}\|{\bm{Z}}^{(\ell,01)}-{\bm{Z}}^{(\ell,10)}-{\bm{Z}}^{(\ell,11)}\|_1
+\|\tilde{\bm{G}}_{ee}^{(\ell-1)}\|_{\infty}\|-{\bm{Z}}^{(\ell,01)}-{\bm{Z}}^{(\ell,10)}+{\bm{Z}}^{(\ell,11)}\|_1
\end{split}
\end{equation*}
Since ${\bm{Z}}^{(\ell,01)}=\frac12((\bm{Z}^{(\ell,01)}+\bm{Z}^{(\ell,11)})+(\bm{Z}^{(\ell,01)}-\bm{Z}^{(\ell,11)}))$, ${\bm{Z}}^{(\ell,10)}=\frac12((\bm{Z}^{(\ell,10)}+\bm{Z}^{(\ell,11)})+(\bm{Z}^{(\ell,10)}-\bm{Z}^{(\ell,11)}))$, and
${\bm{Z}}^{(\ell,11)}=\frac12((\bm{Z}^{(\ell,01)}+\bm{Z}^{(\ell,11)})-(\bm{Z}^{(\ell,01)}-\bm{Z}^{(\ell,11)}))$, we have
\begin{equation*}
\begin{split}
&\|{\bm{Z}}^{(\ell,01)}+{\bm{Z}}^{(\ell,10)}+{\bm{Z}}^{(\ell,11)}\|_1
\leq\frac12(\|\bm{Z}^{(\ell,01)}+\bm{Z}^{(\ell,11)}\|_1+\|\bm{Z}^{(\ell,01)}-\bm{Z}^{(\ell,11)}\|_1)\cr &+\frac12(\|\bm{Z}^{(\ell,10)}+\bm{Z}^{(\ell,11)}\|_1+\|\bm{Z}^{(\ell,10)}-\bm{Z}^{(\ell,11)}\|_1)
+\frac12(\|\bm{Z}^{(\ell,01)}+\bm{Z}^{(\ell,11)}\|_1+\|\bm{Z}^{(\ell,01)}-\bm{Z}^{(\ell,11)}\|_1)\cr
&\leq \frac{3}{2} \frac{4\sqrt2\sqrt{K}}{2^{\ell-1}}
\end{split}
\end{equation*}
and similarly,
$$
\|{\bm{Z}}^{(\ell,01)}-{\bm{Z}}^{(\ell,10)}-{\bm{Z}}^{(\ell,11)}\|_1,
\|-{\bm{Z}}^{(\ell,01)}-{\bm{Z}}^{(\ell,10)}+{\bm{Z}}^{(\ell,11)}\|_1,\|-{\bm{Z}}^{(\ell,01)}+{\bm{Z}}^{(\ell,10)}-{\bm{Z}}^{(\ell,11)}\|_1\leq \frac{3}{2} \frac{4\sqrt2\sqrt{K}}{2^{\ell-1}}
$$
Therefore,
$$
\langle\bm{G},\hat{\bm{Z}}^{(\ell,10)}+\hat{\bm{Z}}^{(\ell,01)}+\hat{\bm{Z}}^{(\ell,11)}\rangle
\leq \frac{3\cdot 4}{2} \frac{4\sqrt{2}\sqrt{K}}{2^{\ell-1}}\|\tilde{\bm{G}}^{(\ell-1)}\|_{\infty}
=8\cdot\sqrt{2}\cdot3 \frac{\sqrt{K}}{2^{\ell-1}}\|\tilde{\bm{G}}^{(\ell-1)}\|_{\infty}
$$

\end{itemize}

Now we are ready to estimate the Gaussian width of $\tilde{\mathcal{S}}_d$. Let $\bm{G}$ be a vector whose entries are i.i.d. Gaussian random variables with mean $0$ and variance $1$. The same argument in one dimensional cases leads to
$$
E(\|\tilde{\bm{G}}^{(\ell)}\|_{\infty})\leq\sqrt{2^{d\ell}2\ln\left({e^{1/2}N^d/2^{d\ell}}\right)}
$$
which implies
\begin{equation*}
\begin{split}
E\left(\sum_{\bm{i}\in\{0,1\}^d\setminus\bm{0}}\langle\bm{G},\hat{\bm{Z}}^{(\ell,\bm{i})}\rangle\right)
&\leq8\sqrt{d}(2^d-1)\frac{\sqrt{K}}{2^{(\ell-1)(d-1)}}E(\|\tilde{\bm{G}}^{(\ell-1)}\|_{\infty})\cr
&\leq8\sqrt{d}(2^d-1)\frac{\sqrt{K}}{2^{(\ell-1)(d-1)}}\sqrt{2^{d(\ell-1)}2\ln\left({e^{1/2}N^d/2^{d(\ell-1)}}\right)}\cr
&\leq8\sqrt{d}(2^d-1)\frac{\sqrt{K}}{2^{(\ell-1)(d-1)}}\sqrt{2^{d(\ell-1)}2\ln\left({e^{1/2}N^d}\right)}\cr
&=8\sqrt{d}(2^d-1)\sqrt{K}2^{(\ell-1)(1-\frac{d}{2})}\sqrt{2\ln\left({e^{1/2}N^d}\right)}.
\end{split}
\end{equation*}
Moreover,
$$
E\left(\langle\bm{G},\hat{\bm{Y}}^{(L)}\rangle\right)
=E\left(Y^{(L)}\tilde{\bm{G}}^{(L)}\right)\leq |Y^{(L)}|E\left(\|\tilde{\bm{G}}^{(L)}\|_{\infty}\right)
\leq \sqrt{\frac{1}{2^{dL}}}\sqrt{2^{dL}2\ln\left({e^{1/2}N^d/2^{dL}}\right)}=\sqrt3.
$$
Therefore,
\begin{equation*}
\begin{split}
E(\langle\bm{G},\bm{X}\rangle)
&=\sum_{\ell=1}^{L-1}E\left(\sum_{\bm{i}\in\{0,1\}^d\setminus\bm{0}}\langle\bm{G},\hat{\bm{Z}}^{(\ell,\bm{i})}\rangle\right)+E\left(\langle\bm{G},\hat{\bm{Y}}^{(L)}\rangle\right)\cr
&=8\sqrt{d}(2^d-1)\sqrt{K}\sqrt{2\ln\left({e^{3/2}N^d}\right)}\sum_{\ell=1}^{L-1}2^{(\ell-1)(1-\frac{d}{2})}+\sqrt3\cr
&\leq
\begin{cases}
8\sqrt{d}(2^d-1)\sqrt{K}\sqrt{2\ln\left({e^{1/2}N^d}\right)}
\log_2N + \sqrt3& \mbox{if~} d=2,\\
8\sqrt{d}(2^d-1)\sqrt{K}\sqrt{2\ln\left({e^{1/2}N^d}\right)}
\frac{2^{1-\frac{d}{2}}}{1-2^{1-\frac{d}{2}}}+\sqrt3 & \mbox{if~} d>2
\end{cases}
\end{split}
\end{equation*}

We require the Gaussian width is about $\sqrt{M}$, where $M$ is the number of measurement. So, we have
\begin{equation*}
M\sim
\begin{cases}
K\log_2^2N\ln N &\mbox{if~} d=2\cr
K\ln N & \mbox{if~} d>2.
\end{cases}
\end{equation*}

\section{Conclusion}
\label{sec:conclusion}

In this paper,we establish the proof for the performance guarantee of total variation (TV) minimization in recovering \emph{one-dimensional} signal with sparse gradient support. The almost Euclidean property of subspaces \cite{Kasin, Yin,nonlinear} is used to extend our results to proving the stability of TV minimization for signals with approximately sparse gradients or under noisy measurements. This partially answers the open problem of proving the fidelity of total variation minimization in such a setting \cite{TVMulti}. We also extend our results to TV minimization for multidimensional signals. Recoverable sparsity thresholds of TV minimization are explicitly computed for $1$-dimensional signal by using the Grassmann angle framework. Stability of TV minimization has also been established for $1$-dimensional signal vectors.

Our current results work only for the Gaussian ensemble of measurement matrices. One future direction is to extend our results to general deterministic and random measurement matrices, such as partial Fourier matrices, and random Bernoulli matrices. Another direction we would like to pursue is to tighten our bounds for $1$-dimensional signal vector. For multidimensional signals, we conjecture that for Gaussian measurement operators, when the number of measurements is proportional to the problem dimension $N^d$, the recoverable sparsity of gradient support, by the TV minimization, can also grow proportionally with $N^d$. We are also interested in working towards tightening our results in this direction.

\bibliographystyle{SIAM}
\bibliography{TVMin}

\end{document}